\pdfoutput=1
\documentclass[letterpaper,12pt]{article}

\usepackage[left=1in, top=1.5in, right=1in, bottom=1.5in]{geometry}
\usepackage{setspace}
\usepackage{appendix}

\usepackage[T1]{fontenc}
\usepackage[charter]{mathdesign}
\usepackage{microtype}
\usepackage{titling}

\pretitle{\begin{center}\large\bfseries}
\posttitle{\end{center}\vspace{-0.5em}}
\preauthor{\begin{center}\large}   
\postauthor{\end{center}\vspace{-0.5em}}
\predate{\begin{center}\small}
\postdate{\end{center}\vspace{-1em}}

\usepackage{csquotes}
\usepackage{enumerate}
\usepackage{bm}
\usepackage{dsfont}

\usepackage{amsmath}
\usepackage{amsthm}

\newtheorem{proposition}{Proposition}
\newtheorem{corollary}{Corollary}
\newtheorem{lemma}{Lemma}

\theoremstyle{definition}
\newtheorem{definition}{Definition}
\newtheorem{assumption}{Assumption}
\theoremstyle{remark}

\usepackage{booktabs}
\usepackage{makecell}
\usepackage{multirow}

\usepackage{graphicx}
\graphicspath{{./figures/}}
\usepackage{xcolor}
\usepackage{caption}
\usepackage{subfigure}
\usepackage{tikz,pgfplots}
\usetikzlibrary{shapes,arrows,quotes}
\pgfplotsset{compat=1.18}

\definecolor{tc1}{RGB}{36, 55, 59}

\usepackage{mdframed}
\mdfdefinestyle{comment}{backgroundcolor = cyan!15, linecolor=cyan, font=\small\sffamily\slshape, fontcolor=black!75}

\usepackage[backend=biber,authordate,url=false,isbn=false,eprint=false]{biblatex-chicago}
\title{When AI Improves Answers but Slows Knowledge Creation: \\Matching and Dynamic Knowledge Creation in Digital Public Goods}
\author{Keh-Kuan Sun
    \thanks{Department of Economics, Fairfield University.  Email: ksun@fairfield.edu.  We thank colleagues at Washington University in St. Louis and Fairfield University, and participants at the 2026 Midwest Economic Theory Conference, the 2026 North American Summer Meeting of the Econometric Society, and the 2026 ESIF Economics and AI+ML Meeting, for helpful comments. We thank Refine.ink, GPT 6.1, Gemini 3.1 Pro, and Claude Opus 5.5 for research assistance with manuscript preparation, including drafting, editing, and literature review. The author retains full responsibility for the content of this article.}
}
\date{November 8, 2026}

\begin{document}

\maketitle
\normalsize
\begin{spacing}{1.0}
\begin{abstract}
Generative AI helps users solve problems more efficiently, but can also replace exchanges between people with private conversations with chatbots. When more problems are solved privately, fewer problem-solving episodes are converted into reusable public knowledge for others. We build a dynamic model of public good provision where agents post and answer questions on a public platform. AI can shrink the platform's archive through two margins. Private diversion can reduce the posted volume of knowledge-enhancing questions (the flow margin), while contributor exit can lower conditional resolution (the resolution margin). Low-archive traps can arise with or without AI, because a thin archive makes public answering too costly to sustain; AI can narrow the viable region. The model gives conditional empirical predictions that distinguish demand-side diversion from supply-side pool thinning under explicit identifying restrictions. Logging private solutions adds archive inflow, while rewards for answering strengthen contribution incentives.
\end{abstract}


\end{spacing}

\vfill
\pagebreak

\section{Introduction}
\label{sec:intro}
Generative AI helps users solve problems more efficiently. Users get faster resolution, lower search and verification costs, and measurable productivity gains in knowledge work \autocite{brynjolfsson2025,dillon2025}. But AI can also replace exchanges with colleagues or other users with private conversations with chatbots. When more problems are solved privately, fewer problem-solving episodes are converted into reusable public knowledge for others. Consider a new hire who brings an unfamiliar problem to a manager. The solution is documented in an internal FAQ that future employees can reuse. A private AI response may solve the employee's problem just as well, but leave no documentation for coworkers. The same principle operates on open knowledge platforms. On Stack Overflow, a posted question and its accepted answer become a public artifact that future visitors can reuse. When the same question is resolved privately, that solution need not reach the public archive. Repeated across many such episodes, individual efficiency gains can coexist with slower accumulation of shared knowledge.

We build a dynamic model of public good provision where agents contribute by solving problems that other agents posted on a public platform, and the accumulated solutions form a depreciating public archive. 
AI can reduce archive creation through two margins. The flow margin is the decline in posted volume of knowledge-enhancing queries; private diversion occurs when AI resolves problems that would otherwise have been posted publicly. The resolution margin is the decline in the probability that posted questions get answered. AI can raise the return to working alone, inducing contributor exit that removes the pool's highest-ability members. When capacity binds, a smaller pool also leaves posted questions less likely to encounter an answerer. Because posted flow and resolution probability are equilibrium objects, the two margins are not causally independent. Private diversion also eases congestion (pushing resolution upward), and pool thinning can feed back into posting incentives. Archive growth can improve AI-assisted private resolution and off-platform returns, creating pressures toward private diversion and contributor exit.

Online Q\&A platforms offer a clean setting for studying these two margins because posting and answer production leave observable public records. 
Evidence of declining public activity and traffic diversion is accumulating across platforms. \textcite{delrio-chanona2024} document a sharp decline in Stack Overflow activity relative to less exposed Q\&A platforms after the release of ChatGPT. \textcite{padilla2025} document declines in traditional search use and visits to Stack Overflow following LLM adoption. Among recently created, popular Wikipedia articles, \textcite{lyu2025} find larger declines in viewership and editing for articles whose content is more similar to ChatGPT-generated content. 
These patterns motivate our central intuition: AI can improve private problem-solving while weakening the public interactions through which reusable knowledge accumulates.
The same organizational problem arises in internal Q\&A systems and enterprise knowledge repositories. Stronger private tools can reduce the return to answering other employees' questions. The organizational question is when logging private solutions can sustain the shared archive and when additional support for public answering is needed. The open-platform setting remains the model's primary application, but the underlying design problem also concerns how organizations preserve cumulative internal knowledge.

The model distinguishes routine questions, which are standard and frequently asked, from knowledge-enhancing questions, whose resolution generates reusable public knowledge. AI improves private problem-solving for both types. By resolving routine questions in private, it can enrich the posted pool in knowledge-enhancing questions. After private resolution fails, agents decide whether to post by comparing the posting cost with the expected value of a public answer. Posted questions are matched to contributors who have limited answering capacity. Contributors differ in ability: abler contributors are more likely to answer a randomly drawn question. Because answering capacity is fixed and contributors commit their discretionary period to one use, AI-freed time raises the opportunity cost of contributing. The domain's routine-task share enters both query composition and the time AI frees from routine work.

The main results fall into three groups. The first concerns the platform within a period. AI's diversion of routine questions can enrich the posted pool in knowledge-enhancing questions, increasing the incentive to answer when contributors value archive expansion. Once contributors can leave, however, resolution depends on four forces: a richer question mix and fewer questions competing for contributors raise it, while a smaller contributor pool and the loss of its ablest answerers lower it. The second group concerns the archive over time. If AI reduces public knowledge creation at the archive size sustained in the human-only environment, the AI economy has a steady state with a smaller archive. The creation shortfall decomposes exactly into the flow and resolution margins. The platform can have multiple steady states, including a structural minimum viable archive below which a thin archive makes answering too costly to sustain creation. This threshold can arise with or without AI, and AI can narrow the viable region. A richer archive can also improve contributors' outside options and create exit pressure. The third group concerns policy. Logging private solutions adds archive inflow, while rewards for answering strengthen contribution incentives.

Section~\ref{sec:model} introduces the model. Section~\ref{sec:main_results} analyzes the flow and resolution margins and their empirical implications. Section~\ref{sec:shrinkage} studies the minimum viable archive, multiple steady states, and contributor exit. Section~\ref{sec:policy} examines conversion and rewards for answering. Section~\ref{sec:discussion} discusses the model's scope and organizational applications and relates our findings to the growing literature on how AI reshapes the economics of knowledge production. Proofs and supporting material are provided in the appendices.

\section{Model}
\label{sec:model}
Time is discrete, $t=0,1,2,\dots$. The platform maintains a public archive stock $K_t\ge 0$, which depreciates at rate $\lambda\in(0,1)$ per period and facilitates both private resolution and the production of public answers. Depreciation is interpreted as effective archive depreciation: obsolescence (APIs and libraries update), link rot, outdated solutions, moderation and closure, and decay in search visibility that reduces reuse of old answers. The core dynamic object is $h(K)-\lambda K$, where $h(K)$ is expected new public knowledge created at state $K$; steady states satisfy $h(K)=\lambda K$ and provide a ``platform viability'' benchmark.

\subsection{Domain and agents}
\label{sec:domain}
The platform serves a knowledge domain in which a fraction $\pi\in(0,1)$ of problems are routine (type~$L$), such as standard lookups, well-documented procedures, and frequently asked questions. Examples include looking up the syntax of a library function or finding the documented procedure for generating a sales report. The remaining fraction $1-\pi$ are knowledge-enhancing (type~$H$), novel problems whose resolution generates reusable public knowledge. Examples include diagnosing a previously undocumented interaction between software libraries or developing a reusable correction for duplicate orders caused by a new sales-to-fulfillment workflow. The parameter $\pi$ is a property of the domain (e.g., introductory programming vs.\ research mathematics), not of individual agents.\footnote{In domains with distinct user and contributor populations, the routine-task shares relevant to composition and outside options may differ; a separate $\pi_c$ for the contributor side accommodates this. The baseline assumes a self-serving community where a single $\pi$ governs both sides, as on Stack Overflow.}

There is a unit mass of agents, each indexed by persistent ability $\alpha_i\sim\Psi$ with density $\psi>0$ on $[\underline\alpha,\bar\alpha]$. Each period, each agent faces either a type~$L$ task with probability $\pi$ or a type~$H$ task with probability $1-\pi$. This generates a unit mass of queries. Within type $\theta\in\{L,H\}$, a query draws residual difficulty $r\sim G_\theta$ with density $g_\theta>0$ on $[r_L,r_H]$. Residual difficulty governs amenability to private resolution (search engines, AI, documentation); it does not determine the cost of answering on the public platform, which depends on the idiosyncratic match quality between question and answerer (Section~\ref{sec:answering}). Each agent is also a potential contributor. After facing their own task, they may be matched to someone else's posted query on the platform.

The routine-task share $\pi$ also determines how much time AI frees from routine work. Agents allocate their discretionary time either to answering others' questions on the platform or to off-platform work. Contribution capacity is normalized to one question per period and remains fixed as AI frees time, whereas off-platform production scales with available time. Appendix~\ref{app:outside_option} formalizes this time allocation and the resulting payoff from off-platform work.

\subsection{Private resolution and endogenous posting}
\label{sec:posting}

Let $e\in\{\mathrm{HO},\mathrm{AI}\}$ index the environment, where $\mathrm{HO}$ denotes the human-only environment and $\mathrm{AI}$ denotes the AI-enabled environment.
In each environment, an agent whose task becomes a query first attempts private resolution using available problem-solving resources. In the human-only environment ($\mathrm{HO}$), these resources include search engines, documentation, coworkers, and other sources of assistance; in the AI-enabled environment ($\mathrm{AI}$), they are augmented with an LLM front end. Private resolution succeeds for a type-$\theta$ query held by an agent of ability $\alpha_i$ whenever
\begin{equation}
r\le \bar r_\theta(\alpha_i,e,K),
\label{eq:private_threshold}
\end{equation}
where $\bar r_\theta(\alpha,e,K)$ is an environment-specific reliability cutoff for private resolution, increasing in $\alpha$, that may depend on the archive stock $K$.\footnote{We interpret $\bar r_\theta$ as an effective reliability cutoff: the hardest query for which the agent relies on private resolution, which can be lower than the hardest query that private tools can solve. The two differ when reliability varies unevenly across tasks and the agent cannot tell in advance whether a given answer is reliable \autocite{gans2026a}. An agent who anticipates hallucination or miscalibration on harder queries does not rely on private answers there and considers posting instead, so higher hallucination or miscalibration risk lowers the effective cutoff.} Because the posting decision below does not condition on $\alpha$, only the cross-agent average private-resolution probability enters posted flows, so we define the aggregate
\begin{equation}
a_\theta^{e}(K)\;\equiv\;\int G_\theta\big(\bar r_\theta(\alpha,e,K)\big)\,d\Psi(\alpha),
\label{eq:a_aggregate}
\end{equation}
and maintain $a_\theta^{\mathrm{AI}}(K)\ge a_\theta^{\mathrm{HO}}(K)$ at the aggregate level for the relevant tasks.
Allowing $a_\theta^{\mathrm{AI}}(K)$ to increase with $K$ captures AI tools that draw on the public archive, either as training data or as content retrieved at query time, so that a richer archive improves private resolution.

If private resolution fails ($r>\bar r_\theta(\alpha_i,e,K)$), the agent decides whether to post the query to the public platform. For these privately unresolved queries, posting entails a cost $d\sim\Gamma$ with CDF $\Gamma:[0,\infty)\to[0,1]$ satisfying $\Gamma(0)=0$, common across query types; this cost captures time to format a question, privacy concerns, stigma, and verification costs \autocite{kummer2020}. Let $V_\theta>0$ denote the private value of obtaining a public answer, with $V_H>V_L$, and let $\sigma^e(K)\in[0,1]$ denote the lifetime resolution probability of a posted query at state $K$ in environment $e$, derived from the answering game below. Type heterogeneity on the posting margin runs entirely through $V_\theta$. The agent's expected benefit from posting is
\begin{equation}
U_\theta^{e}(K)\equiv \sigma^e(K)\cdot V_\theta.
\label{eq:posting_benefit}
\end{equation}
The agent posts iff $U_\theta^{e}(K)\ge d$, implying an endogenous escalation probability
\begin{equation}
m_\theta^{e}(K)\equiv \Gamma\!\big(U_\theta^{e}(K)\big)\in[0,1],
\label{eq:m_endog}
\end{equation}
and therefore a within-type posting rate
\begin{equation}
q_\theta^{e}(K)\equiv m_\theta^{e}(K)\big[1-a_\theta^{e}(K)\big].
\label{eq:q_posted}
\end{equation}
The total posted flow is $Q^e(K)=\pi\,q_L^{e}(K)+(1-\pi)\,q_H^{e}(K)$. Here $q_H^e$ is the fraction of knowledge-enhancing queries that fail private resolution and are posted; their aggregate posted flow is $(1-\pi)q_H^e$.

\subsection{Answering on the public platform}
\label{sec:answering}
A posted query remains active for $T\ge 1$ periods ($T$ finite); the baseline analysis sets $T=1$, so that a query is either resolved within the period it is posted or it expires.\footnote{Appendix~\ref{app:T_extension} proves existence, uniqueness, a closed form for lifetime resolution, and comparative statics for the stationary matching block at fixed contributor pool and answer rate; with a fixed positive pool and answer rate, it also proves a unique active posting outcome when inflow is continuous and nondecreasing in lifetime resolution, vanishes at zero, and is positive otherwise.} Each contributing agent can process at most one query per period. Write the contributor mass as $\Psi(\alpha^{*,e}(K))$, where $\alpha^{*,e}(K)\in[\underline\alpha,\bar\alpha]$ is the upper ability cutoff for participation in environment $e$ at state $K$. Section~\ref{sec:surplus} establishes this cutoff structure and determines the cutoff from the participation decision. Under $T=1$, $Q^e(K)$ is also the stock of active queries. The probability that an active query is matched to a contributor is
\begin{equation}
\mu^e(K)=\min\!\bigg\{1,\;\frac{\Psi(\alpha^{*,e}(K))}{Q^e(K)}\bigg\}.
\label{eq:mu}
\end{equation}
When the contributor pool is large relative to posted flow ($\Psi(\alpha^{*,e})\ge Q^e$), every query is matched and $\mu=1$; when the pool is small ($\Psi(\alpha^{*,e})<Q^e$), queries compete for scarce contributor attention and $\mu<1$. When capacity binds, pool thinning mechanically lowers matching probability at fixed posted flow.\footnote{The Leontief form assumes perfect allocation without search frictions. Under the restrictions in Appendix~\ref{app:matching_general}, replacing it with $M(\Psi,Q)/Q$ for a general matching function preserves the four directional forces.} If no queries are posted ($Q^e(K)=0$), we define $\mu^e(K)=0$.

Conditional on being matched, a contributor of ability $\alpha$ draws an answering cost
\begin{equation}
c\;\sim\;F(\cdot\mid\alpha)\quad\text{on }[0,\bar c],\qquad F(c\mid\alpha)\text{ nondecreasing in }\alpha\text{ for every }c,
\label{eq:F_family}
\end{equation}
with $F(0\mid\alpha)=0$ and $F(c\mid\alpha)$ continuous in $\alpha$. The draws are independent across matches and of the query's type and residual difficulty. Abler contributors draw stochastically lower answering costs. That is, the answer probability $F(c\mid\alpha)$ at any cutoff $c$ is nondecreasing in $\alpha$. For the derivative forms below we impose the following regularity condition.

\begin{assumption}[Regularity of ability-conditioned answering costs]
\label{ass:ability_cost_regularity}
For a.e.\ $c$, the map $\alpha\mapsto F(c\mid\alpha)$ is absolutely continuous on $[\underline\alpha,\bar\alpha]$ with a.e.\ derivative $\frac{\partial F(c\mid\alpha)}{\partial\alpha}$, and the derivative is dominated by an integrable function. Specifically, $|\frac{\partial F(c\mid\alpha)}{\partial\alpha}|\le\vartheta(c)$ wherever it exists, with $\int_0^{\bar c}\vartheta(c)\,dc<\infty$.
\end{assumption}

The answering cost $c$ measures how well a contributor's knowledge, tools, and experience fit the query, distinct from residual difficulty $r$ of private resolution. A question about an older software system may be straightforward yet costly for an engineer unfamiliar with it and cheap for someone who maintained it. Conversely, a statistical problem that is hard to resolve privately may be answered quickly by a researcher who developed the relevant method.

Ability represents expertise breadth. In other words, abler contributors are more likely to answer a randomly matched query at low cost. A narrow specialist can still be the best match for a particular query; such matches are low-cost realizations of the conditional draw, rather than a separate contributor type.

At a given answering cutoff, exit by the highest-ability contributors weakly lowers the pool-average answer rate. We maintain that the answering-cost distribution is invariant to query type and residual difficulty, so that type-dependent difficulty enters private resolution without directly changing public answering costs. Allowing costs to vary with residual difficulty would couple answering to posted difficulty and qualify the empirical comparison in Section~\ref{sec:empirical_implications}.

The matched agent observes the posted query but does not condition their answering decision on its type classification. Under $T=1$, the agent faces a one-shot answering decision with no option value of waiting. Let $C(K)$ denote a state-dependent answering-cost shifter with $C'(K)<0$, capturing that a richer archive makes answering cheaper. At a given archive stock, public answering technology is held fixed across environments. Let $u\ge 0$ denote the private benefit from answering (intrinsic motivation, reciprocity, reputation). Let $\beta\ge 0$ measure the agent's concern for the archive expansion their answer generates. A resolved type-$H$ query adds $\Delta>0$ units of public knowledge. Let $\bar\Delta^e(K)$ denote the expected knowledge increment from answering a randomly drawn posted query:
\begin{equation}
\bar\Delta^e(K)=\frac{(1-\pi)\,q_H^{e}(K)\cdot\Delta}{Q^e(K)}\qquad\text{for }Q^e(K)>0.
\label{eq:delta_bar}
\end{equation}
Thus $\bar\Delta^e(K)$ equals $\Delta$ times the probability that the posted query is type~$H$; when $Q^e(K)=0$, set $\bar\Delta^e(K)=0$. If the agent answers, they incur $C(K)+c$ and receive $\beta\bar\Delta^e(K)+u$. The parameter $\beta$ captures the value contributors attach to adding reusable knowledge to the public archive, reflecting professional norms, community identity, or internalized social value.

The agent answers iff the net benefit is non-negative, yielding an answering cutoff in closed form:
\begin{equation}
c^{*,e}(K)=\max\!\big\{0,\;\beta\bar\Delta^e(K)+u-C(K)\big\}.
\label{eq:cutoff}
\end{equation}
The cutoff $c^{*,e}(K)$ is common to all abilities. A matched contributor of ability $\alpha$ answers with probability $F(c^{*,e}(K)\mid\alpha)$; averaging over the contributing pool $[\underline\alpha,\alpha^{*,e}(K)]$ defines the pool-average answer rate
\begin{equation}
\bar s^{\,e}(K)\;\equiv\;\mathbb{E}\big[F(c^{*,e}(K)\mid\alpha)\,\big|\,\alpha\le\alpha^{*,e}(K)\big]\;=\;\frac{1}{\Psi(\alpha^{*,e}(K))}\int_{\underline\alpha}^{\alpha^{*,e}(K)}F(c^{*,e}(K)\mid\alpha)\,d\Psi(\alpha),
\label{eq:sbar}
\end{equation}
with the convention $\bar s^{\,e}(K)=0$ if $\Psi(\alpha^{*,e}(K))=0$. Under $T=1$, the lifetime resolution probability of a posted query coincides with the per-period resolution hazard:
\begin{equation}
\sigma^e(K)=\mu^e(K)\cdot \bar s^{\,e}(K).
\label{eq:sigma}
\end{equation}
For a given contributor pool, a higher answering cutoff weakly raises the pool-average answer rate. For a given answering cutoff, expanding the pool toward higher abilities also weakly raises this rate. Contributor exit therefore weakly reduces pool-average answering capability, with a strict decline when the departing contributors have a higher average answer probability than those remaining. This quality channel operates whether or not capacity binds. Holding posted flow fixed, contributor exit additionally reduces the query match probability when capacity is constrained.

\subsection{Contribution surplus and participation}
\label{sec:surplus}

\label{sec:participation}
Let $w^e(\alpha,K;\pi)$ denote the expected payoff from allocating the discretionary slot to off-platform work in environment $e$.\footnote{We use this object in reduced form in the main analysis. Appendix~\ref{app:outside_option} derives it from a time-allocation model in which off-platform projects may be routine or knowledge-enhancing.}

\begin{lemma}[Outside-option properties]
\label{lem:outside_option}
Under the off-platform project model in Appendix~\ref{app:outside_option}, the outside option satisfies:
\begin{enumerate}[(a)]
\item \label{lem:outside:ability} $\frac{\partial w^e(\alpha,K;\pi)}{\partial\alpha}>0$;
\item \label{lem:outside:archive} $\frac{\partial w^e(\alpha,K;\pi)}{\partial K}\ge 0$; and
\item \label{lem:outside:ai} $w^{\mathrm{AI}}(\alpha,K;\pi)-w^{\mathrm{HO}}(\alpha,K;\pi)\ge 0$.
\end{enumerate}
\end{lemma}

That is, the outside option increases with ability and weakly with the archive stock; AI weakly raises it at a given ability and archive stock.

A contributing agent is matched to a posted query with per-period probability
\[
\nu^e(K)\equiv\min\{1,\,Q^e(K)/\Psi(\alpha^{*,e}(K))\}.
\] The query-side match probability $\mu^e=\min\{1,\,\Psi/Q^e\}$ measures the probability that a posted query encounters a contributor; $\nu^e$ measures the probability that a contributing agent is matched to a query. Thus, when $Q\ge\Psi$, every participant is matched; when $Q<\Psi$, queries are scarce and not all participants receive one. At the corners we set $\nu^e=0$ if $Q^e=0$, and $\nu^e=1$ if $Q^e>0=\Psi(\alpha^{*,e})$. Conditional on being matched, the expected surplus of a contributor of ability $\alpha$ from the answering opportunity is
\begin{equation}
\Pi^e(\alpha,K)=\int_0^{c^{*,e}(K)}\!\big(c^{*,e}(K)-c\big)\,dF(c\mid\alpha)\;=\;\int_0^{c^{*,e}(K)}F(c\mid\alpha)\,dc,
\label{eq:Pi}
\end{equation}
where the second equality is integration by parts; the surplus equals zero when $c^{*,e}(K)=0$ (no agent answers). Under Assumption~\ref{ass:ability_cost_regularity}, $\frac{\partial\Pi^e(\alpha,K)}{\partial\alpha}=\int_0^{c^{*,e}}\frac{\partial F(c\mid\alpha)}{\partial\alpha}\,dc\ge 0$. Abler contributors expect weakly more surplus from a match. The expected per-period payoff from contribution is
\begin{equation}
S^e(\alpha,K)=\nu^e(K)\cdot\Pi^e(\alpha,K).
\label{eq:S}
\end{equation}
The discretionary slot is allocated before matching is realized, so an unmatched contributor forgoes off-platform work for that period. For example, a contributor may spend time browsing Stack Overflow without finding a question to answer. At the empty-pool corner, a measure-zero contributor facing positive posted flow is matched with certainty. 
The non-excludable archive benefit ($K$ improves everyone's tools) accrues to all agents regardless of contribution and therefore does not enter the contribution comparison. Heterogeneity in the contribution decision is driven by outside options $w^e(\alpha,K;\pi)$ and, through $\Pi^e(\alpha,K)$, by ability-conditioned answering; the single-crossing condition below disciplines their interaction. We maintain that the answering surplus is bounded, with $C(K)\ge\underline C>-\infty$, so $c^{*,e}(K)\le\max\{0,\beta\Delta+u-\underline C\}\equiv\bar c^{\,*}<\infty$ and $\Pi^e(\alpha,K)\le c^{*,e}(K)\le\bar c^{\,*}$ uniformly in $(\alpha,K)$; hence $S^e\le\bar c^{\,*}$ uniformly.

The contribution decision closes the participation margin. Agent $i$ contributes iff $S^e(\alpha_i,K)\ge w^e(\alpha_i,K;\pi)$, that is, iff contributing pays at least as much as off-platform work. Because both platform and off-platform payoffs rise in ability, the interval structure of the contribution set does not follow from monotonicity alone. It is determined by the following single-crossing condition.
\begin{assumption}[Participation single crossing]
\label{ass:participation_sc}
For all $(\alpha,K,e)$,
\begin{equation}
\nu^e(K)\int_0^{c^{*,e}(K)}\frac{\partial F(c\mid\alpha)}{\partial\alpha}\,dc\;<\;\frac{\partial w^e(\alpha,K;\pi)}{\partial\alpha},
\label{eq:participation_sc}
\end{equation}
so $\nu^e\frac{\partial\Pi^e}{\partial\alpha}<\frac{\partial w^e}{\partial\alpha}$ pointwise. Platform returns to ability (through stochastically better match draws) rise more slowly than outside returns to ability.
\end{assumption}

\begin{lemma}[Participation cutoff]
\label{lem:single_crossing}
Fix $(K,e)$ and the matching and answering objects $\nu^e(K)$ and $c^{*,e}(K)$. Under Assumptions~\ref{ass:ability_cost_regularity} and~\ref{ass:participation_sc}, participation is characterized by a unique ability cutoff $\alpha^{*,e}(K)\in[\underline\alpha,\bar\alpha]$: up to a $\Psi$-null boundary agent, contributors are those with $\alpha\le\alpha^{*,e}(K)$. Contributor mass is $\Psi(\alpha^{*,e}(K))$.
\end{lemma}

The proof is in Appendix~\ref{app:outside_option}. Under Assumption~\ref{ass:participation_sc}, the marginal contributor is the pool's ablest member and, by the stochastic ordering in Section~\ref{sec:answering}, answers at a weakly higher rate than the pool average. At a given answering cutoff, exit weakly lowers the pool-average answer rate whether or not capacity binds, with a strict decline when departing contributors have a higher average answer probability than those remaining.

When queries are scarce, a contributor pool that contracts proportionately more than posted flow raises each remaining contributor's chance of a match. Holding match surplus fixed, this partly offsets the incentive to exit. The loss of answering capability can nevertheless lower resolution.

\subsection{Knowledge accumulation and timing}
\label{sec:accumulation}
In either environment, baseline public knowledge accumulates only when a posted knowledge-enhancing (type-$H$) query is resolved on the platform and its solution is publicly logged. Let $h^e(K)$ denote expected new public knowledge at state $K$ in environment $e\in\{\mathrm{HO},\mathrm{AI}\}$. Baseline creation is
\begin{equation}
h^e(K)=\Delta(1-\pi)\,q_H^e(K)\,\sigma^e(K).
\label{eq:h_baseline}
\end{equation}
Thus, public knowledge creation equals the flow of posted knowledge-enhancing queries times their resolution probability, scaled by the knowledge increment $\Delta$. The public archive stock evolves according to
\begin{equation}
K_{t+1}=(1-\lambda)K_t+h^e(K_t).
\label{eq:lom}
\end{equation}
A steady state $K^e$ satisfies $h^e(K^e)=\lambda K^e$. For $K>0$, define the average creation rate $\phi^e(K)\equiv h^e(K)/K$.

The within-period sequence is as follows. First, each agent observes $K_t$ and the environment $e$ and allocates the discretionary slot, contributing iff $S^e(\alpha_i,K_t)\ge w^e(\alpha_i,K_t;\pi)$, which under Lemma~\ref{lem:single_crossing} determines the contributor pool $\Psi(\alpha^{*,e}(K_t))$. Next, each agent receives a task as a problem owner, attempts private resolution, and decides whether to post, determining posted flows $q_L^e$, $q_H^e$, and total flow $Q^e$. Active queries are then matched to contributing agents under the congestion technology; matched agents draw $c\sim F(\cdot\mid\alpha_i)$ and answer iff $c\le c^{*,e}(K_t)$, and resolved queries are publicly logged. Finally, resolved type-$H$ queries contribute $\Delta$ to $K$, and the archive updates to $K_{t+1}$.

\subsection{Period equilibrium}
\label{sec:equilibrium}
At a given $(K,e)$, a period equilibrium jointly determines contribution, posting, matching, answering, and participation. It is active if questions are posted ($Q^e(K)>0$) and inactive if none are posted ($Q^e(K)=0$). An inactive equilibrium always exists. We select an active equilibrium whenever one exists and use zero creation otherwise. The baseline permits multiple active equilibria; comparisons refer to a specified selection, and the dynamic results require selected knowledge creation to vary continuously with the archive stock on the intervals used in their statements.

Posting incentives depend on expected answers, while the chance of obtaining an answer depends on posted queries and the contributor pool. Appendix~\ref{app:remarks} gives sufficient conditions for an active equilibrium to exist, for it to be unique, and for selected knowledge creation to vary continuously with the archive stock. A steady state additionally requires $h^e(K^e)=\lambda K^e$.

\section{AI and the Two Margins}
\label{sec:main_results}
We now turn to the main comparative statics. We focus on type-$H$ queries, because only $H$ contributes to the archive stock. 
By \eqref{eq:h_baseline}, creation depends on aggregate knowledge-enhancing posted flow $(1-\pi)q_H^e(K)$ and conditional resolution $\sigma^e(K)$. At fixed $\pi$, flow comparisons can use $q_H^e$. Posting responds to resolution, and resolution depends on flow and contributor mass through matching. This is an accounting decomposition, with both margins linked through equilibrium feedback. Private diversion acts on posting; contributor exit can reduce answering capability and, when capacity binds, matching.
We begin with a composition result that isolates one channel through which AI can improve resolution, and then show how contributor exit can overwhelm this gain.

\subsection{Composition of the posted pool}
\label{sec:composition}
The answering cutoff $c^{*,e}(K)=\max\{0,\,\beta\bar\Delta^e(K)+u-C(K)\}$ from \eqref{eq:cutoff} depends on the environment $e$ only through the expected knowledge increment $\bar\Delta^e(K)$, which in turn depends on the $H$-share of the posted pool. Changes in posted-pool composition can change the incentive to answer, even holding the contributor pool fixed. The following result isolates this composition channel.

\begin{lemma}[Composition of the posted pool and the answering cutoff]
\label{lem:composition}
Fix $K$ such that total posted flow is positive in both environments. Let $\tilde c^{e}(K)\equiv\beta\bar\Delta^e(K)+u-C(K)$ denote the latent answering cutoff, so that $c^{*,e}(K)=\max\{0,\tilde c^{e}(K)\}$. When $\beta>0$, $\tilde c^{\mathrm{AI}}(K)\ge \tilde c^{\mathrm{HO}}(K)$ if and only if the posted pool is weakly more $H$-rich under AI:
\begin{equation}
q_H^{\mathrm{AI}}(K)q_L^{\mathrm{HO}}(K)\ge q_H^{\mathrm{HO}}(K)q_L^{\mathrm{AI}}(K).
\label{eq:h_rich}
\end{equation}
\end{lemma}

\begin{corollary}[A sufficient condition for composition enrichment]
\label{cor:composition_fixed_escalation}
Fix $K$ and hold escalation probabilities fixed, $m_\theta^{\mathrm{AI}}=m_\theta^{\mathrm{HO}}\equiv m_\theta>0$ for both types. Suppose $a_\theta^{\mathrm{HO}}(K)<1$ for both types. If
\begin{equation}
\frac{1-a_H^{\mathrm{AI}}(K)}{1-a_H^{\mathrm{HO}}(K)}>\frac{1-a_L^{\mathrm{AI}}(K)}{1-a_L^{\mathrm{HO}}(K)},
\label{eq:diff_capability}
\end{equation}
then the posted pool is more $H$-rich under AI. When $\beta>0$, the latent answering cutoff is higher and the realized answering cutoff is weakly higher under AI.
\end{corollary}

The condition requires the private-failure rate to decline proportionally more for routine than for knowledge-enhancing queries.

The corollary holds escalation fixed to isolate AI's direct composition effect. In equilibrium, expected resolution also affects posting incentives, so composition must be evaluated at the endogenous posted flows as in Lemma~\ref{lem:composition}.

Lemma~\ref{lem:composition} connects to the selection logic emphasized by \textcite{gans2024}. By diverting routine queries to private resolution, AI can enrich the posted pool in knowledge-enhancing problems and raise the incentive to answer. The composition channel here operates through the answering cutoff $c^{*}$ (intensive margin), isolated from a separate congestion channel that operates through the match probability $\mu$ (extensive margin); Proposition~\ref{prop:resolution} examines their interaction. The composition channel through the answering cutoff requires both $\beta>0$ and answering decisions that do not condition on query type. When agents do not value archive expansion ($\beta=0$), the answering cutoff $c^{*}=\max\{0,\,u-C(K)\}$ is independent of the posted-pool composition, and the composition channel is absent.

\subsection{Resolution: composition, congestion, and pool quality}
\label{sec:resolution}
Lemma~\ref{lem:composition} isolates the effect of posted-pool composition on willingness to answer. Whether greater willingness to answer translates into higher resolution depends on matching capacity and the contributor pool's answering capability. Contributor exit can weaken answering capability even without congestion and reduce matching when capacity binds. The following proposition separates these contributor-side effects from composition and congestion relief.

\begin{proposition}[Resolution: composition, congestion, and pool quality]
\label{prop:resolution}
Fix $K$. For the pool average define
\begin{equation*}
\bar s(c,a)\equiv\frac{1}{\Psi(a)}\int_{\underline\alpha}^{a}F(c\mid\alpha)\,d\Psi(\alpha).
\end{equation*}
Suppose both environments are in active equilibria with positive contributor mass ($\Psi(\alpha^{*,e}(K))>0$), both resolving ($\sigma^e(K)>0$), and the platform congested in both, so $\Psi(\alpha^{*,e}(K))<Q^e(K)$ for $e\in\{\mathrm{HO},\mathrm{AI}\}$. Under $T=1$, $\sigma^e(K)=\big(\Psi(\alpha^{*,e}(K))/Q^e(K)\big)\cdot\bar s\big(c^{*,e}(K),\alpha^{*,e}(K)\big)$. Then, whenever the anchor $\bar s(c^{*,\mathrm{AI}},\alpha^{*,\mathrm{HO}})$ is positive (automatic if $\alpha^{*,\mathrm{HO}}\ge\alpha^{*,\mathrm{AI}}$), the ratio factors exactly:
\begin{equation}
\frac{\sigma^{\mathrm{AI}}}{\sigma^{\mathrm{HO}}}\;=\;\underbrace{\frac{\Psi(\alpha^{*,\mathrm{AI}})}{\Psi(\alpha^{*,\mathrm{HO}})}}_{\text{pool size}}\times\underbrace{\frac{Q^{\mathrm{HO}}}{Q^{\mathrm{AI}}}}_{\text{congestion relief}}\times\underbrace{\frac{\bar s(c^{*,\mathrm{AI}},\alpha^{*,\mathrm{AI}})}{\bar s(c^{*,\mathrm{AI}},\alpha^{*,\mathrm{HO}})}}_{\text{pool quality}}\times\underbrace{\frac{\bar s(c^{*,\mathrm{AI}},\alpha^{*,\mathrm{HO}})}{\bar s(c^{*,\mathrm{HO}},\alpha^{*,\mathrm{HO}})}}_{\text{composition}},
\label{eq:resolution_race}
\end{equation}
with all arguments evaluated at $K$. Resolution falls under AI ($\sigma^{\mathrm{AI}}<\sigma^{\mathrm{HO}}$) iff the product is less than one.
\end{proposition}

The factorization compares realized equilibria. Posting responds to expected resolution, while participation depends on platform surplus relative to off-platform returns; the magnitudes of the four factors therefore depend on these equilibrium responses.

The four factors in \eqref{eq:resolution_race} determine the change in resolution. Two operate on the query side. When AI enriches the posted pool in knowledge-enhancing queries, the higher expected knowledge increment can raise willingness to answer, weakly increasing the answer rate of a given contributor pool (composition, Lemma~\ref{lem:composition}). Lower posted flow gives each question a better chance of meeting a contributor on a congested platform (congestion relief).

Two operate on the contributor side. AI raises outside options (Lemma~\ref{lem:outside_option}(\ref{lem:outside:ai})), which can induce exit. A smaller pool lowers matching probability on a congested platform (pool size). Because the ablest contributors leave, exit also weakly lowers the pool-average answer rate at a given answering cutoff, whether or not capacity binds (pool quality). Both contributor-side forces follow from a lower participation cutoff. Resolution falls when their combined effect outweighs composition and congestion relief.

\subsection{Dynamic crowd-out}
\label{sec:crowdout}
We now use the two margins jointly to characterize long-run archive outcomes. At a human-only steady state $K^{\mathrm{HO}}$, creation exactly offsets depreciation, $h^{\mathrm{HO}}(K^{\mathrm{HO}})=\lambda K^{\mathrm{HO}}$. If AI reduces creation at that point, that archive size cannot be sustained.

\begin{proposition}[Dynamic crowd-out and lower archive steady states]
\label{prop:dynamic_crowdout}
Let $K^{\mathrm{HO}}>0$ be a steady state in the human-only economy, so that $h^{\mathrm{HO}}(K^{\mathrm{HO}})=\lambda K^{\mathrm{HO}}$. Suppose at $K^{\mathrm{HO}}$ the AI economy satisfies
\begin{equation}
h^{\mathrm{AI}}(K^{\mathrm{HO}})<h^{\mathrm{HO}}(K^{\mathrm{HO}}),
\label{eq:h_decline}
\end{equation}
i.e., AI reduces one-step knowledge creation at the human-only steady state. If selected $h^{\mathrm{AI}}(\cdot)$ is continuous on $[0,K^{\mathrm{HO}}]$, the AI economy admits a steady state $K^{\mathrm{AI}}\in[0,K^{\mathrm{HO}})$.
\end{proposition}

Proposition~\ref{prop:dynamic_crowdout} is an existence result, and the AI steady state it delivers may be the empty archive, $K^{\mathrm{AI}}=0$. If AI creation exceeds depreciation at some positive stock below the human-only steady state, continuity guarantees an AI steady state strictly between those two stocks.
The result applies separately to each human-only steady state, whether stable or unstable.

The following corollary decomposes the condition in Proposition~\ref{prop:dynamic_crowdout} into the two margins and connects them to the channels identified in Proposition~\ref{prop:resolution}.

\begin{corollary}[Two-margin decomposition of crowd-out]
\label{cor:decomp}
The condition $h^{\mathrm{AI}}(K^{\mathrm{HO}})<h^{\mathrm{HO}}(K^{\mathrm{HO}})$ is equivalent to
\begin{equation}
q_H^{\mathrm{HO}}(K^{\mathrm{HO}})\,\sigma^{\mathrm{HO}}(K^{\mathrm{HO}})>q_H^{\mathrm{AI}}(K^{\mathrm{HO}})\,\sigma^{\mathrm{AI}}(K^{\mathrm{HO}}).
\label{eq:corollary_cond}
\end{equation}
With all terms evaluated at $K^{\mathrm{HO}}$, define $\Delta q_H\equiv q_H^{\mathrm{HO}}-q_H^{\mathrm{AI}}$, $\Delta\sigma\equiv \sigma^{\mathrm{HO}}-\sigma^{\mathrm{AI}}$, $\bar q_H\equiv(q_H^{\mathrm{HO}}+q_H^{\mathrm{AI}})/2$, and $\bar\sigma\equiv(\sigma^{\mathrm{HO}}+\sigma^{\mathrm{AI}})/2$. The decline decomposes as
\begin{equation}
q_H^{\mathrm{HO}}\sigma^{\mathrm{HO}}-q_H^{\mathrm{AI}}\sigma^{\mathrm{AI}}=\underbrace{\bar\sigma\,\Delta q_H}_{\text{flow margin}}+\underbrace{\bar q_H\,\Delta\sigma}_{\text{resolution margin}}.
\label{eq:two_margin_decomp}
\end{equation}

\end{corollary}

The decomposition is an accounting identity, not a causal decomposition. A decline in the flow component ($\Delta q_H>0$) can reflect private diversion as well as equilibrium feedback from resolution and congestion. A decline in the resolution component ($\Delta\sigma>0$) is the net outcome of pool shrinkage, pool-quality thinning, composition change, and congestion relief.

\subsection{Empirical implications of the two margins}
\label{sec:empirical_implications}

The model provides a conditional diagnostic of contributor-side thinning. At a given archive stock, with public answering technology and incentives held fixed or controlled, the restrictions below imply that a fixed contributor pool has unchanged or higher conditional resolution when total posted flow $Q$ falls; that is, lower flow weakly eases congestion, and a fixed or more $H$-rich mix weakly raises its answering probability.

Contributor exit can reduce answering capability and, when capacity binds, matching. A joint decline in total posted flow and conditional resolution therefore points to contributor-side thinning. When every posted query is matched in both environments, a resolution decline points to the pool-quality channel.

This diagnostic requires three identifying restrictions. First, the demand-side prediction requires the per-query match rate $M(\Psi,Q)/Q$ to decrease in $Q$ at fixed $\Psi$; with increasing returns to queue length, lower demand need not improve matching. Second, posted-pool composition must be held fixed or become weakly more $H$-rich, as in Lemma~\ref{lem:composition}. If AI diverts $H$-type queries more strongly than $L$-type queries, the pool becomes less $H$-rich and the answering cutoff can fall without contributor exit. Third, empirical work must condition on question difficulty. The baseline keeps $F(\cdot\mid\alpha)$ invariant to residual difficulty, whereas a shift toward questions that are harder for humans to answer could lower conditional resolution without exit.

As discussed in the introduction, existing studies document declines in public posting and search use \autocite{delrio-chanona2024,padilla2025}. \textcite{quinn2025} also document changes in casual users' posted questions: after ChatGPT's release, questions became more novel and more complex by readability measures. \textcite{li2025} add evidence on contributor behavior: the decline in undeleted Stack Overflow answers was concentrated among previously highly active contributors. Connecting this finding to the pool-thinning mechanism requires measuring contributor exit and conditional resolution.

Testing this diagnostic requires platform-by-time data on total posted flow $Q_{p,t}$ and conditional resolution $\sigma_{p,t}$ around an AI capability shock, a less-exposed comparison group, and uncertainty estimates. Contributor histories and expertise measures track participation and capability; question-level measures control composition and difficulty. Under these controls, joint declines in total posted flow and conditional resolution with an unchanged contributor pool would contradict the diagnostic; exit concentrated among less capable contributors would contradict the model's sorting pattern.

To examine cross-platform differences, the routine-task share could be measured through expert coding or validated algorithmic classification of a representative sample of pre-AI problems, including those not posted publicly. The classification would distinguish established procedures and lookups from problems whose solutions add reusable knowledge. Classifying posted queries before and after an AI capability shock would then allow comparisons of the change in knowledge-enhancing-to-routine odds across domains. Under the conditions in Section~\ref{sec:pi_vulnerability}, our account predicts a common multiplicative shift.

\section{Archive Dynamics and Contributor Exit}
\label{sec:shrinkage}
Proposition~\ref{prop:dynamic_crowdout} links a reduction in creation at a human-only steady state to the existence of a lower archive steady state. This section examines archive-dependent outside options and gives sufficient conditions for a minimum viable archive and multiple steady states.

A thin archive can make answering too costly to sustain creation. As the archive grows, answering becomes cheaper, but off-platform returns can also improve and draw contributors away. Under the conditions below, creation exceeds depreciation at an intermediate archive size but falls short at low and high stocks, yielding multiple steady states. The structural threshold is not AI-specific, but AI can narrow the viable region through the two margins identified in the previous section.

\subsection{Archive-dependent outside options}
\label{sec:feedback}
The contribution decision balances platform and off-platform returns. At an interior contribution cutoff,
\[
\nu^e(K)\,\Pi^e(\alpha^{*,e}(K),K)=w^e(\alpha^{*,e}(K),K;\pi).
\]
The marginal contributor is the pool's ablest member under Assumption~\ref{ass:participation_sc}. A richer archive affects this balance through two channels. Holding the posted question mix fixed, lower answering costs ($C'(K)<0$) weakly raise the answering cutoff $c^{*,e}(K)$ and the expected surplus from an answering opportunity, supporting contribution. At the same time, a richer archive can improve off-platform returns ($\frac{\partial w^e}{\partial K}\ge 0$), creating exit pressure. The balance between these channels depends on the archive stock and the responses of posting, matching, and question composition.

A partial stabilizer operates when contributors outnumber posted queries ($\Psi>Q$), so some contributors remain unmatched. If the pool contracts proportionately more than posted flow, each remaining contributor's match probability $\nu^e=\min\{1,\,Q/\Psi\}$ rises. Holding match surplus fixed, this raises the participation payoff $S^e(\alpha,K)$ and partly offsets exit pressure.

At a given answering cutoff, exit still weakly lowers average answering capability, so resolution can fall even when every posted query is matched. Once queries outnumber contributors ($\Psi<Q$), every contributor is already matched with certainty, and the matching stabilizer is exhausted.

\subsection{Multiple steady states}
\label{sec:steady_states}
The following result gives sufficient conditions for multiple steady states in the AI economy. Recall that, for $K>0$, average creation per unit of archive stock is $\phi^{\mathrm{AI}}(K)=h^{\mathrm{AI}}(K)/K$.

\begin{samepage}
\begin{assumption}[Conditions for multiple steady states]
\label{ass:phi_shape}
\leavevmode
\begin{enumerate}[(a)]
\item \textbf{Empty-archive shutdown:} $C(0)>\beta\Delta+u$.
\item \textbf{Finite-threshold full exit:} $w_\infty\equiv\lim_{K\to\infty}w^{\mathrm{AI}}(\underline\alpha,K;\pi)\;>\;\max\{0,\beta\Delta+u-\underline C\}=\bar c^{\,*}$.
\item \textbf{Selected-creation continuity:} Use the active-first convention of Section~\ref{sec:equilibrium} and require selected $h^{\mathrm{AI}}$ to be continuous on $[0,\infty)$.
\item \textbf{Intermediate viability:} There exists $K_0>0$ such that $h^{\mathrm{AI}}(K_0)>\lambda K_0$.
\end{enumerate}
\end{assumption}
\end{samepage}

Condition~(a) means that when the archive is empty, the answering-cost shifter is so high that no agent is willing to answer even the most valuable query: $\beta\bar\Delta^e(0)+u-C(0)\le\beta\Delta+u-C(0)<0$, so $c^{*,e}(0)=0$.

Condition~(b) means that the lowest-ability agent's outside option eventually exceeds the uniform surplus bound of Section~\ref{sec:surplus}, so even the least able agent eventually exits. The limit $w_\infty$ exists in $[0,\infty]$ because $w^{\mathrm{AI}}(\underline\alpha,\cdot)$ is nondecreasing by Lemma~\ref{lem:outside_option}(\ref{lem:outside:archive}), and it may be infinite. An unbounded outside option is sufficient but not required; the displayed strict inequality is a convenient sufficient condition, not a necessary one.

Condition~(c) makes $\phi^{\mathrm{AI}}$ continuous on $(0,\infty)$, and condition~(d) requires the platform to be viable at some intermediate archive size. These are substantive assumptions that need justification. Appendix~\ref{app:period_equilibrium} provides sufficient conditions for both. In particular, low answering costs, participation payoffs competitive with off-platform work, and enough knowledge-enhancing queries left unresolved privately can support creation above depreciation.\footnote{Proposition~\ref{prop:global_continuity} supplies sufficient conditions for condition~(c), without requiring active existence at every stock. Lemma~\ref{lem:primitive_viability} gives sufficient conditions for intermediate viability when the selected active equilibrium's cutoff lies in its certified compact domain.}

Together, the four conditions place $\phi^{\mathrm{AI}}(K)$ below $\lambda$ near zero, above $\lambda$ at the intermediate stock $K_0$, and below $\lambda$ for sufficiently large $K$. Continuity therefore gives at least two positive steady states.

\begin{proposition}[Multiple steady states]
\label{prop:shrinkage}
Under Assumption~\ref{ass:phi_shape}, the AI economy has at least two positive steady states $0<K_U^{\mathrm{AI}}<K_H^{\mathrm{AI}}$, where $U$ and $H$ label the smallest and largest positive steady states, respectively. For an initial stock $0<K_0<K_U^{\mathrm{AI}}$, the archive remains positive at every finite date, declines monotonically, and converges to zero as $t\to\infty$; if $K_0=0$, it remains at zero. If, in addition, $h^{\mathrm{AI}}$ is differentiable at these steady states and
\[
h^{\mathrm{AI}\prime}(K_U^{\mathrm{AI}})>\lambda>h^{\mathrm{AI}\prime}(K_H^{\mathrm{AI}})>\lambda-2,
\]
then $K_U^{\mathrm{AI}}$ is a locally unstable threshold and $K_H^{\mathrm{AI}}$ is locally asymptotically stable.
\end{proposition}

The multiple-steady-state structure is not AI-specific. The empty-archive shutdown in condition~(a) depends on $C(0)>\beta\Delta+u$, which is environment-independent. The human-only economy has the same structure under the analogous existence and stability conditions, as in the parametric example. Its lower and higher crossings are denoted $K_U^{\mathrm{HO}}$ and $K_H^{\mathrm{HO}}$. AI can narrow the viable region from both ends. From the right, if AI reduces creation at every archive level at or above the human-only stable state, the stable steady state compresses inward ($K_H^{\mathrm{AI}}<K_H^{\mathrm{HO}}$): creation then falls short of depreciation throughout that range, so every positive AI steady state lies below $K_H^{\mathrm{HO}}$. From the left, the stronger AI outside option can bind participation near the activation point, raising the structural minimum viable archive ($K_U^{\mathrm{AI}}>K_U^{\mathrm{HO}}$). Figure~\ref{fig:parametric} illustrates this margin in the parametric example: the up-crossing runs through the participation switch-on.

\begin{figure}[t]
\centering
\includegraphics[width=\textwidth]{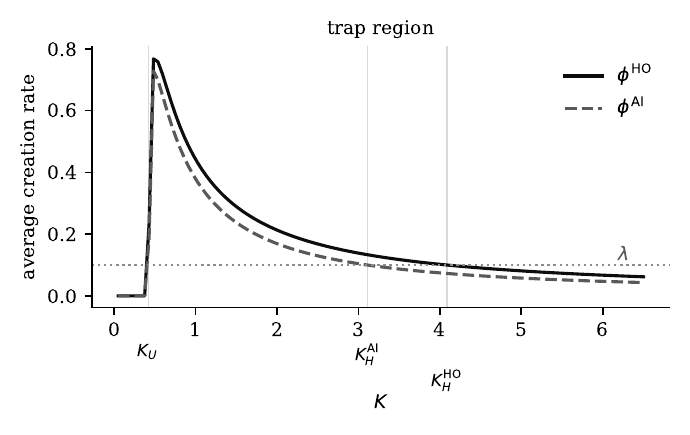}
\caption{How AI can narrow the viable region: an uncalibrated worked example using the specification in Appendix~\ref{app:parametric}. The figure plots the average creation rate in the human-only and AI economies against depreciation $\lambda$ (dotted). Under AI, the lower crossing rises from approximately $0.415$ to $0.420$, while the upper crossing falls from approximately $4.09$ to $3.11$.}\label{fig:parametric}
\end{figure}

High answering costs make very small archives unviable.\footnote{A natural question is whether, if the archive declines toward a low-$K$ region, the process can be reversed by reducing reliance on AI. In practice this is not a frictionless margin. AI is an off-platform outside option that individual users and contributors cannot unilaterally remove, and platform-wide reduction in AI use would require coordination or external regulation. Moreover, when $K$ is low, answering costs $C(K)$ are high, resolution is low, and few agents find it worthwhile to post or contribute. Even if some agents attempted to rely more on the platform, the thin contributor pool and poor resolution would discourage them, slowing recovery. Endogenizing AI reliance is a natural extension, but the baseline isolates the externality in the open-archive channel.} In Figure~\ref{fig:parametric}, AI raises the minimum viable archive by tightening participation near activation. It also lowers the upper steady state, where participation remains full in both economies. At the human-only upper steady-state stock, the creation shortfall in this example is almost entirely the flow component; conditional resolution is slightly higher under AI.

\subsection{Cross-platform differences: the role of $\pi$}
\label{sec:pi_vulnerability}
The routine-task share $\pi$ offers one account of cross-platform differences, for example if Stack Overflow has a higher routine-task share than MathOverflow. Compare domains differing only in routine-task share, with common type-specific private-resolution rates at a common archive stock and type-specific escalation probabilities held fixed across domains and environments. When both types have positive posted flows, let $O^e\equiv(1-\pi)q_H^e/(\pi q_L^e)$ denote the odds that a posted query is knowledge-enhancing. Then
\[
\frac{O^{\mathrm{AI}}}{O^{\mathrm{HO}}}
=
\frac{(1-a_H^{\mathrm{AI}})(1-a_L^{\mathrm{HO}})}
     {(1-a_H^{\mathrm{HO}})(1-a_L^{\mathrm{AI}})}.
\]
If AI lowers the private-failure rate proportionally more for routine than for knowledge-enhancing queries, as in \eqref{eq:diff_capability}, this multiplier exceeds one and is independent of $\pi$. A high-$\pi$ domain starts from a more routine posted pool, so AI reduces its posted traffic proportionally more and raises its knowledge-enhancing share proportionally more.

A domain difference in AI's relative capability would produce a similar traffic pattern without any difference in $\pi$. \textcite{delrio-chanona2024} document a decline in Stack Overflow activity relative to comparison platforms that include MathOverflow and Mathematics Stack Exchange, and they attribute the difference to ChatGPT's weaker performance in mathematics; the observed difference is therefore consistent with both accounts. Changes in posted-pool composition could help distinguish them (Section~\ref{sec:empirical_implications}). Composition enrichment can counteract contributor-side pressure from AI-freed time, so the fixed-escalation comparison does not establish a general ranking of equilibrium knowledge-creation losses across domains.

\section{Policy Instruments for the Two Margins}
\label{sec:policy}
The previous sections established that AI can reduce archive creation through the flow and resolution margins. This section studies public logging of private solutions and rewards for public answers. Logging adds shared artifacts, while rewards strengthen answering incentives. Whether logging alone can restore pre-AI creation depends on the shortfall relative to the private resolutions available for deposit (Section~\ref{sec:eta_limits}).
Here ``public'' means shared, searchable, and reusable within the relevant modeled community. An organization-wide internal archive is therefore public in the model's economic sense because its benefits are non-excludable within that population, even when the archive remains proprietary outside the organization.

\subsection{Public logging and conversion}

Let $\eta\in[0,1]$ denote the fraction of privately resolved type-$H$ queries in the AI environment that are converted into shared artifacts. Each converted artifact adds the same $\Delta$ units of public knowledge as a publicly resolved $H$-query. Relative to baseline creation in \eqref{eq:h_baseline}, conversion gives
\begin{equation}
h^{\mathrm{AI}}(K;\eta)=h^{\mathrm{AI}}(K)+\Delta(1-\pi)\eta\,a_H^{\mathrm{AI}}(K).
\label{eq:h_ai}
\end{equation}
Thus $\eta=0$ recovers the baseline. Possible implementations include default sharing of AI-assisted solutions, structured archiving, ``post your solved query'' prompts, and institutional logging requirements. In practice, conversion involves both detecting a potentially reusable type-$H$ solution and depositing it in the shared archive. Disclosure policies can support detection, while verification and archiving support deposition.

For example, an LLM assisting a user on Stack Overflow can flag a solution as potentially reusable when it finds no existing answer and prompt the user to verify and share it. Enterprise systems can apply the same process to internal Q\&A boards and knowledge repositories, with deposition by the user or a designated reviewer. The artifact already exists in the chat, but curation is needed before it enters the archive. Curated deposits may later serve as training or evaluation data under a separate governance decision; $\eta$ measures conversion into the shared archive rather than automatic training on undisclosed private interactions.

Under conversion rate $\eta$, the average creation rate in the AI economy shifts to
\begin{equation}
\phi^{\mathrm{AI}}(K;\eta)=\phi^{\mathrm{AI}}(K)+\Delta(1-\pi)\eta\,\frac{a_H^{\mathrm{AI}}(K)}{K}.
\label{eq:phi_shift}
\end{equation}
The added term is the per-unit-of-$K$ contribution of logged private resolutions in the AI environment. For any fixed positive conversion rate, it is bounded below by a positive constant for sufficiently small $K$ (because $a_H^{\mathrm{AI}}(K)/K$ is bounded away from zero near the origin; see the premise of Proposition~\ref{prop:conversion}), and it vanishes as $K\to\infty$ because $a_H^{\mathrm{AI}}$ is bounded while $K$ grows.
This asymmetry is what gives conversion its force at the low-archive margin. 

We maintain that conversion operates without within-period behavioral feedback.

\begin{assumption}[Myopic conversion]
\label{ass:myopic_conversion}
The conversion rate $\eta$ enters only during end-of-period accumulation, after contribution, posting, and answering decisions have been made. Agents do not condition current-period decisions on $\eta$, and do not anticipate the effect of logging on future archive stocks $K_{t+1}$, hence on future outside options $w^e$ and answering costs $C(K)$.
\end{assumption}

At a given archive stock, conversion leaves current posting, participation, and answering decisions unchanged and adds logged private resolutions to public creation. As the archive grows, subsequent decisions respond to the new stock through private problem-solving, answering costs, and off-platform returns.

\begin{proposition}[Conversion and expansion of the viable region]
\label{prop:conversion}
Maintain the conditions of Proposition~\ref{prop:shrinkage} for $\phi^{\mathrm{AI}}(K)$. Suppose $a_H^{\mathrm{AI}}(K)$ is continuous on $(0,\infty)$, strictly positive for all $K>0$, and that $a_H^{\mathrm{AI}}(K)/K$ is bounded below by a positive constant for sufficiently small $K$.
Then:
\begin{enumerate}[(i)]
\item For any $\eta>0$, $\phi^{\mathrm{AI}}(K;\eta)>\phi^{\mathrm{AI}}(K)$ for all $K>0$. Conversion strictly raises average creation.
\item There exists a finite critical conversion rate $\bar\eta>0$ such that $\phi^{\mathrm{AI}}(K;\eta)>\lambda$ for every $K\in(0,K_U^{\mathrm{AI}}]$ whenever $\eta>\bar\eta$. At $\eta=\bar\eta$, this pointwise strict inequality holds if and only if the supremum in \eqref{eq:eta_bar} is not attained. The critical rate is
\begin{equation}
\bar\eta=\sup_{K\in(0,K_U^{\mathrm{AI}}]} \frac{\max\{\lambda-\phi^{\mathrm{AI}}(K),0\}}{\Delta(1-\pi)\,a_H^{\mathrm{AI}}(K)/K}.
\label{eq:eta_bar}
\end{equation}
\end{enumerate}
\end{proposition}

\begin{corollary}[Feasible conversion and uniform growth]
\label{cor:conversion_feasibility}
Under the conditions of Proposition~\ref{prop:conversion}, an implementable rate $\eta\in[0,1]$ yielding $\phi^{\mathrm{AI}}(K;\eta)>\lambda$ for every $K\in(0,K_U^{\mathrm{AI}}]$ exists if and only if $\bar\eta<1$, or $\bar\eta=1$ and the supremum in \eqref{eq:eta_bar} is not attained. For $\eta\in[0,1]$, a strictly positive uniform margin,
\[
\inf_{K\in(0,K_U^{\mathrm{AI}}]}[\phi^{\mathrm{AI}}(K;\eta)-\lambda]>0,
\]
holds if and only if $\eta>\bar\eta$.
\end{corollary}

The critical rate $\bar\eta$ in \eqref{eq:eta_bar} is the worst-case ratio of the baseline creation shortfall to conversion effectiveness. When an implementable conversion rate exceeds $\bar\eta$, it eliminates the low-archive basin of collapse: the archive grows at every positive stock up to the original minimum viable archive $K_U^{\mathrm{AI}}$.

Figure~\ref{fig:eta_shifts} illustrates the upward shift in average creation. If even full conversion cannot cover the creation shortfall, additional measures that raise creation are needed.

\begin{figure}[t]
\centering
\includegraphics[width=0.75\textwidth]{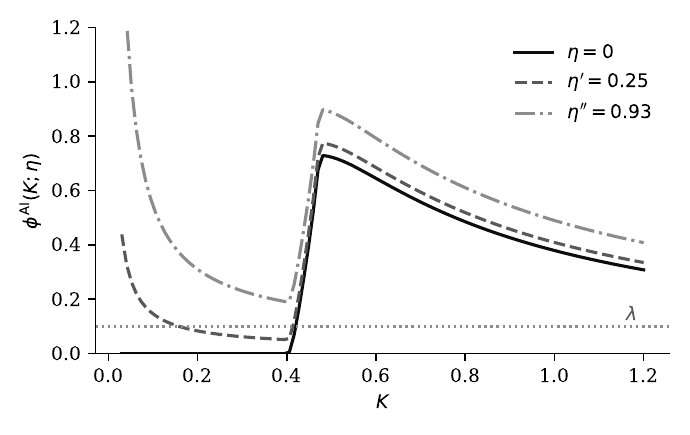}
\caption{How conversion can eliminate the low-archive basin of collapse: an uncalibrated worked example using the specification in Appendix~\ref{app:parametric}. The figure plots the average creation rate in the AI economy against depreciation $\lambda$ (dotted) for no ($\eta=0$), moderate ($\eta=0.25$), and high ($\eta=0.93$) conversion. Moderate conversion raises average creation but leaves it below $\lambda$ in part of the region below the minimum viable archive; high conversion gives positive drift at every sampled stock below the baseline threshold. The sampled critical conversion rate is approximately $0.49$.}\label{fig:eta_shifts}
\end{figure}

\subsection{Platform-resolution rewards}
A complementary instrument targets the resolution margin by raising the return to producing an answer on the platform. Let $\zeta\ge 0$ denote a type-neutral reward paid only when a matched contributor produces an answer. The reward may be monetary or nonmonetary. Because contributors do not condition their answering decision on the $L/H$ classification, $\zeta$ shifts the common answering benefit from $u$ to $u+\zeta$ without changing the baseline information structure. At fixed posted-pool composition and contributor pool, the answering cutoff and conditional contribution surplus become
\[
c^{*,e}(K;\zeta)=\max\{0,\,\beta\bar\Delta^e(K)+u+\zeta-C(K)\},
\qquad
\Pi^e(\alpha,K;\zeta)=\int_0^{c^{*,e}(K;\zeta)}F(c\mid\alpha)\,dc.
\]
Holding the question mix, posted flow, and contributor pool fixed, a higher reward weakly raises willingness to answer and contribution surplus. Its equilibrium effect also depends on how posting, question composition, and participation respond.

Within an organization, $\zeta$ can represent compensation, performance credit, or recognition awarded for producing reusable answers. The same interpretation applies to contributions made through internal Q\&A boards, documentation systems, and workflows that curate problem--solution pairs for later training or evaluation.

The two instruments also differ in timing: an answer reward affects participation and answering, whereas conversion enters end-of-period accumulation under Assumption~\ref{ass:myopic_conversion}. 

\subsection{Combined use, implementation, and limits}
\label{sec:eta_limits}
At a given archive stock, conversion adds private resolutions to the archive while leaving posting and conditional resolution unchanged. Preventing archive decline requires deposits to cover any shortfall of creation without conversion relative to depreciation. Restoring human-only creation requires covering any shortfall relative to the human-only creation rate at the same stock. Whether conversion alone achieves either objective depends on the available private resolutions relative to the corresponding shortfall. If even full conversion cannot cover it, an additional source of creation is needed; rewards for answering are one possible instrument.

Lastly, the model does not claim that increasing $\eta$ is costless or unambiguously welfare improving. Real implementations may entail privacy costs (some private queries contain sensitive information), intellectual property costs (firms may not want problem-solving strategies public), curation costs (unfiltered AI outputs can degrade the archive rather than sustain it), and compliance or legislative costs. The model isolates the archive-stock externality and examines conversion as a source of archive inflow; optimal policy design that balances these costs against archive preservation is left to future work.

\section{Discussion}
\label{sec:discussion}
AI's private gains can coexist with a dynamic loss of shared platform knowledge. By diverting users' problems to private resolution and raising contributors' returns to off-platform work, AI can reduce the shared archive available to future users.

Our article shares with \textcite{acemoglu2026} the concern that private gains from AI can erode shared knowledge. Their model studies human learning and includes extensions with AI aggregation and synthetic public information. We study how private problem-solving and contributor participation affect the posting and resolution of public questions. This distinction motivates logging private solutions and supporting public answering as policy instruments. Their analysis characterizes welfare and optimal information design, including garbling AI recommendations to sustain human learning; ours identifies the separate platform margins through which knowledge creation changes.

\textcite{ide_talamas2025} study how AI reshapes the organization of knowledge work within hierarchical firms. Our model examines voluntary participation in public problem solving and the accumulation of a shared archive of reusable solutions.

When participation contracts, the baseline cutoff structure removes the pool's highest-ability members and, at a given answering cutoff, weakly lowers its average answer rate (Section~\ref{sec:answering}). \textcite{kim2026} document reduced uploads among incumbent illustrators on Pixiv after the launch of generative AI, with the largest declines among those with the highest pre-AI posting frequency. This finding motivates testing whether reductions in activity also change participation and the remaining pool's answering capability. The model still abstracts from answer quality conditional on answering, since every resolved $H$-query adds the same $\Delta$. Public answering capacity and technology are held fixed at a given archive stock and question mix; direct AI assistance to public contributors could offset the modeled losses. \textcite{shankar2024} emphasize that AI-assisted public contributions still require scarce expert validation, a governance constraint our model abstracts from.

Contributors who leave a platform may continue creating knowledge in proprietary or academic channels. Whether platform losses translate into an economy-wide decline depends on the importance of the affected channels for cumulative knowledge creation and on whether diverted effort produces equivalent public value elsewhere. The two-margin framework identifies the platform-level changes that a cross-channel investigation would need to relate to knowledge creation elsewhere.

The model offers a starting point for examining how problem solving in these other settings generates knowledge that others can reuse. Its core elements are the stochastic arrival of problems at a knowledge frontier, matching between problems and solvers, a public archive that accumulates from resolved problems and lowers future resolution costs, an outside option competing for solvers' attention, and AI improving that outside option while diverting problems to private resolution. These elements apply, with different incentive structures, to academic research (where the archive is the published literature and depreciation reflects methodological obsolescence), publicly funded policy analysis (where the archive is the stock of reports, briefs, and public datasets), and private R\&D (where the archive is partly proprietary and returns are internalized through intellectual property).

The model's fit varies across these settings, and two dimensions of variation matter most. The first is the incentive to share reusable solutions. Academic priority and tenure norms structurally attenuate the flow margin by making ``posting'' (publication) partially decoupled from private problem resolution, whereas open Q\&A platforms and wikis lack such structural incentives and are correspondingly more exposed. An internal knowledge system can also fall within the model when solutions are shared and reusable across the organization, even if they are proprietary outside it.

The second dimension is how work is allocated. Centrally coordinated research assignments reduce reliance on voluntary contributor availability, which shapes the resolution margin in the present model. Voluntary internal Q\&A boards, documentation systems, and cross-team knowledge-sharing repositories differ: problem owners and potential solvers remain dispersed, so capacity constraints and contribution incentives can still govern whether a reusable answer is produced. Private R\&D with fully assigned tasks and completely internalized returns remains a limiting case in which the externality structure is fundamentally different.

Enterprise governance can support voluntary internal knowledge systems through disclosure and repository integration for conversion, and compensation and recognition for public answering. Protected work time can change opportunity costs or available answering capacity. Implementation remains subject to privacy, intellectual-property, verification, and curation constraints. These settings offer a concrete application of the framework: studying how sharing rules and contribution incentives sustain reusable knowledge alongside private AI assistance.

\pagebreak
\printbibliography
\pagebreak
\appendix

\section{Proofs}
\label{app:dynamics_proofs}
\label{app:policy_results}

\subsection{Proof of Lemma~\ref{lem:composition} and Corollary~\ref{cor:composition_fixed_escalation}}

Write $\omega^e=(1-\pi)q_H^e/Q^e$, so the expected knowledge increment is $\bar\Delta^e=\Delta\omega^e$. With positive total posted flows,
\[
\omega^{\mathrm{AI}}-\omega^{\mathrm{HO}}
=\frac{\pi(1-\pi)}{Q^{\mathrm{AI}}Q^{\mathrm{HO}}}
\big(q_H^{\mathrm{AI}}q_L^{\mathrm{HO}}-q_H^{\mathrm{HO}}q_L^{\mathrm{AI}}\big).
\]
The prefactor is positive, so enrichment is equivalent to the cross-product inequality \eqref{eq:h_rich}. For $\beta>0$, the latent cutoff has the same ordering. Truncation at zero preserves the weak ordering, with equality possible when both latent cutoffs are non-positive.

For the corollary, substitute $q_\theta^e=m_\theta(1-a_\theta^e)$. Multiplying \eqref{eq:diff_capability} by its positive human-only denominators gives the strict cross-product inequality. The premises ensure positive human-only flow, and the strict inequality implies $q_H^{\mathrm{AI}}>0$. Both shares are therefore defined, and the enrichment and cutoff conclusions follow. \qed

\subsection{Proof of Proposition~\ref{prop:resolution}}

Pool-average monotonicity. For $\Psi(a)>0$, the pool average in \eqref{eq:sbar} is nondecreasing in $c$ because each $F(\cdot\mid\alpha)$ is a CDF. Ability ordering gives, wherever $\psi$ exists,
\[
\frac{\partial\bar s(c,a)}{\partial a}=\frac{\psi(a)}{\Psi(a)}[F(c\mid a)-\bar s(c,a)]\ge0.
\]
Both monotonicities are weak; cost-support gaps can leave answer probabilities unchanged.

Factorization. Under joint congestion and positive contributor mass, with positive resolution in both environments,
\[
\begin{aligned}
\sigma^e&=\frac{\Psi(\alpha^{*,e})\bar s(c^{*,e},\alpha^{*,e})}{Q^e},\\
\frac{\sigma^{\mathrm{AI}}}{\sigma^{\mathrm{HO}}}
&=\frac{\Psi(\alpha^{*,\mathrm{AI}})}{\Psi(\alpha^{*,\mathrm{HO}})}
\frac{Q^{\mathrm{HO}}}{Q^{\mathrm{AI}}}
\frac{\bar s(c^{*,\mathrm{AI}},\alpha^{*,\mathrm{AI}})}{\bar s(c^{*,\mathrm{HO}},\alpha^{*,\mathrm{HO}})}.
\end{aligned}
\]
Multiplying and dividing by the positive anchor $\bar s(c^{*,\mathrm{AI}},\alpha^{*,\mathrm{HO}})$ gives \eqref{eq:resolution_race}. When $\alpha^{*,\mathrm{HO}}\ge\alpha^{*,\mathrm{AI}}$, its positivity follows directly from
\[
\bar s(c^{*,\mathrm{AI}},\alpha^{*,\mathrm{HO}})
\ge\bar s(c^{*,\mathrm{AI}},\alpha^{*,\mathrm{AI}})>0.
\]

Conditional signs. The pool-size and pool-quality factors are at most one if $\alpha^{*,\mathrm{AI}}\le\alpha^{*,\mathrm{HO}}$. Congestion relief is at least one if $Q^{\mathrm{AI}}\le Q^{\mathrm{HO}}$. Composition is at least one if $c^{*,\mathrm{AI}}\ge c^{*,\mathrm{HO}}$, as Lemma~\ref{lem:composition} implies under $H$-enrichment and $\beta>0$. Resolution falls exactly when the product is below one. Off joint congestion, $\mu^{\mathrm{AI}}/\mu^{\mathrm{HO}}$ replaces the first two factors; the quality and composition factors remain valid under the same positive-anchor condition. \qed

\subsection{Proof of Proposition~\ref{prop:dynamic_crowdout}}

Let $\varphi(K)=h^{\mathrm{AI}}(K)-\lambda K$, continuous on $[0,K^{\mathrm{HO}}]$ by premise. Nonnegative creation gives $\varphi(0)\ge0$, while \eqref{eq:h_decline} and the human-only steady-state identity give $\varphi(K^{\mathrm{HO}})<0$. If $\varphi(0)=0$, the empty archive is the required lower AI steady state. Otherwise the intermediate value theorem gives a zero in $(0,K^{\mathrm{HO}})$. \qed

\subsection{Proof of Corollary~\ref{cor:decomp}}

Canceling the common positive factor $\Delta(1-\pi)$ in \eqref{eq:h_baseline} makes the creation comparison equivalent to the stated product inequality. Using the corollary's midpoint definitions,
\[
\begin{aligned}
\bar\sigma\Delta q_H+\bar q_H\Delta\sigma
&=\tfrac12(\sigma^{\mathrm{HO}}+\sigma^{\mathrm{AI}})(q_H^{\mathrm{HO}}-q_H^{\mathrm{AI}})
 +\tfrac12(q_H^{\mathrm{HO}}+q_H^{\mathrm{AI}})(\sigma^{\mathrm{HO}}-\sigma^{\mathrm{AI}})\\
&=q_H^{\mathrm{HO}}\sigma^{\mathrm{HO}}-q_H^{\mathrm{AI}}\sigma^{\mathrm{AI}}.
\end{aligned}
\]
The cross terms cancel, leaving no residual interaction term. \qed

\subsection{Proof of Proposition~\ref{prop:shrinkage}}

Shutdown at low and high stocks. Assumption~\ref{ass:phi_shape}(a), continuity of $C$, and $\bar\Delta^e\le\Delta$ make the answering cutoff zero on $[0,\varepsilon]$ for some $\varepsilon>0$. Since $F(0\mid\alpha)=0$, resolution and creation vanish there, so $\phi^{\mathrm{AI}}<\lambda$ on $(0,\varepsilon]$.

The bounded-surplus condition of Section~\ref{sec:surplus} gives $S^e\le\Pi^e\le\bar c^{\,*}$. Under Assumption~\ref{ass:phi_shape}(b), the nondecreasing outside option from Lemma~\ref{lem:outside_option} has limit $w_\infty>\bar c^{\,*}$. Thus at all sufficiently large stocks the lowest-ability agent strictly prefers off-platform work. Lemma~\ref{lem:single_crossing} gives zero contributor mass, hence zero resolution and creation. The strict limit inequality is essential: equality need not imply exit at any finite stock.

Roots and drift. Let $K_v>0$ be the viable intermediate stock in Assumption~\ref{ass:phi_shape}(d), so $\phi^{\mathrm{AI}}(K_v)>\lambda$. Continuity from condition (c) and the low/high shutdown regions give nonempty compact root sets on $[\varepsilon,K_v]$ and $[K_v,\bar K]$, with $\bar K>K_v$ chosen in the high-stock shutdown region. Their smallest and largest roots give $0<K_U^{\mathrm{AI}}<K_v<K_H^{\mathrm{AI}}$. There are no roots below $K_U^{\mathrm{AI}}$ or above $K_H^{\mathrm{AI}}$; continuity and the shutdown signs imply negative drift in both regions. Downward drift above the upper root does not by itself establish discrete-time attraction.

Collapse below the smallest root. For $0<K_{\mathrm{init}}<K_U^{\mathrm{AI}}$, \eqref{eq:lom} and $0<\lambda<1$ give
\[
0<(1-\lambda)K_t\le K_{t+1}<K_t<K_U^{\mathrm{AI}}.
\]
The trajectory decreases to $\ell\ge0$ while remaining positive at every finite date. If $\ell>0$, continuity makes it a fixed point, contradicting the absence of a root below $K_U^{\mathrm{AI}}$. Hence $K_t\to0$. The initial condition $K_{\mathrm{init}}=0$ stays at zero; differentiability is not required for this collapse conclusion.

Conditional local stability. Suppose $h^{\mathrm{AI}}$ is differentiable at $K_U^{\mathrm{AI}}$ and $K_H^{\mathrm{AI}}$, with $h^{\mathrm{AI}\prime}(K_U^{\mathrm{AI}})>\lambda>h^{\mathrm{AI}\prime}(K_H^{\mathrm{AI}})$. These inequalities give transversal, isolated crossings. For the law of motion \eqref{eq:lom}, write $\mathcal G(K)=(1-\lambda)K+h^{\mathrm{AI}}(K)$. At a positive fixed point $K^*$,
\[
\frac{\mathcal G(K)-K^*}{K-K^*}
=(1-\lambda)+\frac{h^{\mathrm{AI}}(K)-h^{\mathrm{AI}}(K^*)}{K-K^*}
\longrightarrow L\equiv1-\lambda+h^{\mathrm{AI}\prime}(K^*).
\]
If $|L|<1$, choose $q$ with $|L|<q<1$. In a sufficiently small neighborhood contained in the positive-stock domain, $|\mathcal G(K)-K^*|\le q|K-K^*|$. This neighborhood is invariant and $|K_t-K^*|\le q^t|K_{\mathrm{init}}-K^*|$, proving local asymptotic stability. If $|L|>1$, choose $c$ with $1<c<|L|$; the corresponding local lower bound $|\mathcal G(K)-K^*|\ge c|K-K^*|$ forces every nontrivial nearby orbit to leave the neighborhood, proving instability. This direct argument requires differentiability only at the root; the standard continuously differentiable version is given in \textcite[Thm.~1.13]{elaydi2005}. At $K_U^{\mathrm{AI}}$, $L>1$, so the threshold is locally unstable. At $K_H^{\mathrm{AI}}$, local asymptotic stability follows if additionally $h^{\mathrm{AI}\prime}(K_H^{\mathrm{AI}})>\lambda-2$, equivalently $|L|<1$. When $L<-1$, nearby displacements alternate sides and expand until leaving the local neighborhood. The equality cases $|L|=1$ are not decided by the derivative alone.
 \qed

\subsection{Proof of Proposition~\ref{prop:conversion} and Corollary~\ref{cor:conversion_feasibility}}

Under myopic conversion, \eqref{eq:phi_shift} raises average creation by $\eta b(K)$ without changing current choices at fixed stock. Since $b(K)>0$, the shift is strict for every positive rate and stock.

Finite critical rate. Set $S=(0,K_U^{\mathrm{AI}}]$,
\[
b(K)=\Delta(1-\pi)\frac{a_H^{\mathrm{AI}}(K)}{K},\qquad
r(K)=\frac{\max\{\lambda-\phi^{\mathrm{AI}}(K),0\}}{b(K)},\qquad
\bar\eta=\sup_{K\in S}r(K).
\]
By the assumed lower bound on $a_H^{\mathrm{AI}}(K)/K$ near zero, $b$ is bounded away from zero on a sufficiently small interval $(0,\varepsilon]$, with $0<\varepsilon<K_U^{\mathrm{AI}}$. Continuity and strict positivity of $b$ on $[\varepsilon,K_U^{\mathrm{AI}}]$ then give $b_*\equiv\inf_S b>0$. Because $0\le\phi^{\mathrm{AI}}\le\lambda$ on $S$ and $\phi^{\mathrm{AI}}<\lambda$ in its interior, $0<\bar\eta\le\lambda/b_*<\infty$.

For every $K\in S$,
\[
\phi^{\mathrm{AI}}(K;\eta)-\lambda=b(K)[\eta-r(K)].
\]

For $\eta>\bar\eta$, this margin is positive throughout $S$; for $\eta<\bar\eta$, it fails to be positive somewhere in $S$.

Pointwise growth and implementability. At $\eta=\bar\eta$, strict growth on $S$ holds exactly when the supremum of $r$ is unattained. Intersecting the successful rates with $[0,1]$ gives an implementable rate exactly when $\bar\eta<1$, or $\bar\eta=1$ with unattained supremum.

Uniform average-creation margin. A feasible positive uniform margin exists exactly when $\bar\eta<1$. For the margin itself, $\eta>\bar\eta$ is sufficient since
\[
\inf_{K\in S}[\phi^{\mathrm{AI}}(K;\eta)-\lambda]
\ge b_*(\eta-\bar\eta)>0.
\]
At $\eta=\bar\eta$, an attained supremum gives equality at its maximizer. If it is unattained, choose $K_n\in S$ with $r(K_n)\to\bar\eta$. Eventually $r(K_n)\ge\bar\eta/2$, so
\[
b(K_n)=\frac{\lambda-\phi^{\mathrm{AI}}(K_n)}{r(K_n)}
\le\frac{2\lambda}{\bar\eta},\qquad
0\le\phi^{\mathrm{AI}}(K_n;\bar\eta)-\lambda
\le\frac{2\lambda}{\bar\eta}[\bar\eta-r(K_n)]\to0.
\]
Hence the uniform margin is zero at the threshold, whether or not the supremum is attained, and its strict positivity is equivalent to $\eta>\bar\eta$. Finally, $\mathcal G_\eta(K)-K=K[\phi^{\mathrm{AI}}(K;\eta)-\lambda]$ establishes the stated drift property on $S$. It does not establish global convergence. \qed

\section{Outside Options and Participation}
\label{app:outside_option}

\subsection{Off-platform projects}

Off-platform projects can be routine or knowledge-enhancing, but their outputs do not enter the platform archive. They are distinct from an agent's own problem and have a separate difficulty distribution $\tilde G$ with density $\tilde g>0$. Per unit of time, a project draws $r'\sim\tilde G$, is completed if $r'\le\bar r_O(\alpha,e,K)$, and yields value $v(K)>0$. We maintain
\[
\frac{\partial\bar r_O}{\partial\alpha}>0,\qquad \frac{\partial\bar r_O}{\partial K}\ge0,\qquad
\bar r_O(\alpha,\mathrm{AI},K)\ge\bar r_O(\alpha,\mathrm{HO},K),\qquad v'(K)\ge0.
\]

Routine obligations use $\pi\tau_R^e$ of the unit period, with $\tau_R^{\mathrm{AI}}<\tau_R^{\mathrm{HO}}$ and $1-\pi\tau_R^{\mathrm{HO}}>0$. The remaining time is
\begin{equation}
\mathcal T^e(\pi)=1-\pi\tau_R^e=\bar\tau(\pi)+\pi\delta^e,\qquad
\bar\tau(\pi)\equiv1-\pi\tau_R^{\mathrm{HO}},\quad
\delta^e\equiv\tau_R^{\mathrm{HO}}-\tau_R^e,
\label{eq:slot}
\end{equation}
where $\delta^{\mathrm{HO}}=0<\delta^{\mathrm{AI}}$ is AI's time dividend. At the start of the period, the whole discretionary slot is committed to platform contribution or off-platform work (Section~\ref{sec:participation}). Contribution capacity is normalized to one query per period, regardless of slot length (Section~\ref{sec:answering}); off-platform work is divisible and scales linearly with time. Its expected payoff is
\begin{equation}
w^e(\alpha,K;\pi)=v(K)\tilde G\big(\bar r_O(\alpha,e,K)\big)\mathcal T^e(\pi).
\label{eq:w_micro}
\end{equation}
Holding $\tilde G$ and $\bar r_O$ fixed with respect to $\pi$ makes routine share affect this payoff through available time.

\subsection{Derivation of Lemma~\ref{lem:outside_option}}

Write $\bar r_O^e=\bar r_O(\alpha,e,K)$. Differentiating \eqref{eq:w_micro} gives
\[
\begin{aligned}
\frac{\partial w^e}{\partial\alpha}&=v(K)\tilde g(\bar r_O^e)\frac{\partial\bar r_O^e}{\partial\alpha}\mathcal T^e(\pi)>0,\\
\frac{\partial w^e}{\partial K}&=\big[v'(K)\tilde G(\bar r_O^e)+v(K)\tilde g(\bar r_O^e)\frac{\partial\bar r_O^e}{\partial K}\big]\mathcal T^e(\pi)\ge0.
\end{aligned}
\]
The AI shift decomposes exactly as
\begin{equation}
w^{\mathrm{AI}}-w^{\mathrm{HO}}
=v(K)\Big\{
\underbrace{\big[\tilde G(\bar r_O^{\mathrm{AI}})-\tilde G(\bar r_O^{\mathrm{HO}})\big]\bar\tau(\pi)}_{\text{capability effect}}
+
\underbrace{\tilde G(\bar r_O^{\mathrm{AI}})\pi\delta^{\mathrm{AI}}}_{\text{time-reallocation effect}}
\Big\}\ge0.
\label{eq:w_shift}
\end{equation}
The capability effect raises off-platform productivity; the time-reallocation effect increases the time available for that work.

\subsection{Proof of Lemma~\ref{lem:single_crossing}}

Fix $(K,e)$ and the selected matching and answering objects $\nu,c^*$. Let $D(\alpha)=\nu\Pi(\alpha;c^*)-w^e(\alpha,K;\pi)$.

Monotonicity. Assumption~\ref{ass:ability_cost_regularity} and Fubini's theorem applied to \eqref{eq:Pi} imply absolute continuity of $\Pi$ in ability and
\[
\Pi(\alpha_2;c^*)-\Pi(\alpha_1;c^*)
=\int_{\alpha_1}^{\alpha_2}\int_0^{c^*}\frac{\partial F(c\mid a)}{\partial a}\,dc\,da.
\]
The outside option is continuous and increasing by Lemma~\ref{lem:outside_option}. Its increment is at least the integral of its derivative, so Assumption~\ref{ass:participation_sc} gives
\[
D(\alpha_2)-D(\alpha_1)
\le\int_{\alpha_1}^{\alpha_2}[\nu\frac{\partial\Pi(a;c^*)}{\partial a}-\frac{\partial w^e(a,K;\pi)}{\partial a}]\,da<0.
\]
Thus $D$ is continuous and strictly decreasing, without requiring absolute continuity of $w$.

Cutoff and corners. If $D(\underline\alpha)\le0$, contributor mass is zero; a possible boundary contributor has measure zero. If $D(\bar\alpha)\ge0$, every agent contributes and mass is one. Otherwise $D(\underline\alpha)>0>D(\bar\alpha)$, so the intermediate value theorem and strict decrease give a unique interior root. Contributors have $\alpha\le\alpha^*$, up to the null boundary. \qed

\section{Period Equilibrium}
\label{app:remarks}
\label{app:period_equilibrium}

\subsection{Period-equilibrium construction}

At fixed $(K,e)$, a period equilibrium jointly satisfies:
\begin{enumerate}[(E1)]
\item Matching probability. $\mu^e(K)=\min\{1,\,\Psi(\alpha^{*,e}(K))/Q^e(K)\}$, with $\mu^e=0$ if $Q^e=0$.
\item Posted-pool composition. $\bar\Delta^e(K)=(1-\pi)\,q_H^{e}(K)\cdot\Delta/Q^e(K)$ if $Q^e>0$; $\bar\Delta^e=0$ if $Q^e=0$.
\item Answering cutoff. $c^{*,e}(K)=\max\{0,\;\beta\bar\Delta^e(K)+u-C(K)\}$.
\item Pool-average answer rate and resolution.\\* $\bar s^{\,e}(K)=\Psi(\alpha^{*,e}(K))^{-1}\int_{\underline\alpha}^{\alpha^{*,e}(K)}F(c^{*,e}(K)\mid\alpha)\,d\Psi(\alpha)$, with $\bar s^{\,e}=0$ if $\Psi(\alpha^{*,e})=0$; and $\sigma^e(K)=\mu^e(K)\cdot\bar s^{\,e}(K)$.
\item Posting benefit and escalation. $U_\theta^{e}(K)=\sigma^e(K)\cdot V_\theta$ and $m_\theta^{e}(K)=\Gamma(U_\theta^{e}(K))$ for $\theta\in\{L,H\}$.
\item Posted flows. $q_\theta^{e}(K)=m_\theta^{e}(K)\cdot[1-a_\theta^{e}(K)]$ and $Q^e(K)=\pi\,q_L^{e}(K)+(1-\pi)\,q_H^{e}(K)$.
\item Contribution surplus. $\nu^e(K)=\min\{1,\,Q^e(K)/\Psi(\alpha^{*,e}(K))\}$ with the corner conventions of Section~\ref{sec:surplus}, $\Pi^e(\alpha,K)=\int_0^{c^{*,e}(K)}F(c\mid\alpha)\,dc$, and $S^e(\alpha,K)=\nu^e(K)\cdot\Pi^e(\alpha,K)$.
\item Contribution cutoff. $\alpha^{*,e}(K)$ solves $S^e(\alpha^{*,e}(K),K)=w^e(\alpha^{*,e}(K),K;\pi)$, with the corner cases described in the proof of Lemma~\ref{lem:single_crossing}, under Assumption~\ref{ass:participation_sc}.
\end{enumerate}
Given a contributor cutoff, the inner loop determines posting, composition, answering, and matching; the outer loop imposes participation condition (E8). A steady state also satisfies $h^e(K)=\lambda K$.

\begin{definition}[Active and inactive period equilibria]
\label{def:active}
An equilibrium is active if $Q^e(K)>0$ and inactive if $Q^e(K)=0$. The inactive profile has $\sigma^e=\mu^e=\bar\Delta^e=S^e=0$ and satisfies (E1)--(E8) at every $(K,e)$ because $\Gamma(0)=0$. We select an active equilibrium whenever one exists and the inactive outcome otherwise. With multiple active equilibria, comparative results use a specified active selection. Dynamic results require continuous selected creation unless a sufficient-condition theorem certifies it. The convention supplies no general welfare ranking and does not exclude coordination multiplicity among active outcomes.
\end{definition}

\subsection{Composition feedback and active multiplicity}
\label{app:multiplicity}
\label{app:multiplicity_econ}

At fixed $(K,e,a)$, set $q_\theta(\sigma)=\Gamma(\sigma V_\theta)[1-a_\theta^e(K)]$, $Q=\pi q_L+(1-\pi)q_H$, and $\omega=(1-\pi)q_H/Q$. On $Q>0$, the inner map is
\[
c^*(\sigma)=\max\{0,\beta\Delta\omega(\sigma)+u-C(K)\},\qquad
\hat\sigma(\sigma)=\min\{1,\Psi(a)/Q(\sigma)\}\bar s(c^*(\sigma),a),
\]
with the pool average from \eqref{eq:sbar}. Active solutions are its positive fixed points. The inactive convention sets resolution to zero at $Q=0$, although the map's active right limit can be positive. Compact-domain existence is addressed below; uniqueness is treated in Lemma~\ref{lem:uniqueness}.

Where $\Gamma$ admits a density $\gamma$ and both type flows are positive, define
\[
\ell_\theta(\sigma)=\frac{V_\theta\gamma(\sigma V_\theta)}{\Gamma(\sigma V_\theta)},\qquad
\omega'(\sigma)=\omega(1-\omega)(\ell_H-\ell_L).
\]
With $\beta>0$, a positive gap can reinforce composition: an $H$-rich pool raises willingness to answer, and higher resolution can induce disproportionately more $H$ posting. This is a possible feedback, not a demonstration of multiplicity or stability. Because both types share $\Gamma$, a nonincreasing map $x\gamma(x)/\Gamma(x)$ and $V_H>V_L$ imply $\ell_H\le\ell_L$, making composition corrective or neutral. The exponential posting distribution in Appendix~\ref{app:parametric} gives a negative gap.

Active multiplicity is a within-period coordination problem at fixed fundamentals, distinct from archive dynamics. Different selected active branches can generate different current creation and therefore different archive paths.

\begin{lemma}[Sufficient condition for inner-loop uniqueness]\label{lem:uniqueness}
Fix $(\alpha^*,K,e)$. Suppose $\Gamma$ admits a density $\gamma$ with $\Gamma>0$ on
$(0,\infty)$ (so the semi-elasticities are defined on the active domain). Suppose each
$F(\cdot\mid\alpha)$ admits a density $f(\cdot\mid\alpha)$ with
$\bar f\equiv\sup_{c,\,a}\;\mathbb E\big[f(c\mid\alpha)\,\big|\,\alpha\le a\big]<\infty$
(the supremum of the pool-average density over cutoffs $c$ and pools $a$). Suppose the maps below
are differentiable where evaluated, and $\hat\sigma$ is absolutely continuous on compact
subintervals of the active region (so the slope bound integrates across branch and clamp
boundaries). If
\begin{equation}\label{eq:uniqueness_cond}
\beta\Delta\cdot\bar f\cdot\tfrac14\,
\sup_{\sigma\in(0,1)}\big[\ell_H(\sigma)-\ell_L(\sigma)\big]^{+}\;<\;1,
\end{equation}
then $\hat\sigma'(\sigma)<1$ wherever defined, and the inner loop has at most one active fixed
point at $(\alpha^*,K,e)$.
\end{lemma}

\begin{proof}
On differentiable parts of the active domain, the product rule gives
\[
\hat\sigma'=\mu'\bar s+\mu\frac{\partial\bar s}{\partial c} c^{*\prime}
\le\bar f\beta\Delta\omega(1-\omega)\max\{\ell_H-\ell_L,0\}
\le\frac{\bar f\beta\Delta}{4}\sup_\sigma\max\{\ell_H-\ell_L,0\}<1.
\]
Here $\mu'\le0$ because posting increases with resolution, and $\mu\le1$. On a binding zero clamp, $c^{*\prime}=0$; if only one type posts, composition is constant. Absolute continuity integrates this bound across matching and clamp boundaries, making $\hat\sigma(\sigma)-\sigma$ strictly decreasing on the positive-posting region. Thus there is at most one active inner root.
\end{proof}

Participation consistency remains to be imposed through the outer residual.

\subsection{Participation consistency}

For a selected inner branch at fixed $(K,e)$, let
\[
\Phi(a)=\nu(a)\Pi(a;c^*(a))-w(a,K;\pi).
\]
\begin{lemma}[Existence of a participation fixed point]
\label{lem:participation_fixed_point}
Fix $(K,e)$. Under Assumptions~\ref{ass:ability_cost_regularity} and~\ref{ass:participation_sc}, suppose the selected inner-loop maps $a\mapsto\nu(a)$ and $a\mapsto c^{*}(a)$ are continuous on $[\underline\alpha,\bar\alpha]$, using the corner conventions of Section~\ref{sec:surplus}. Then a self-consistent cutoff satisfying equilibrium condition (E8) exists.
\end{lemma}

\begin{proof}
The selected $\nu$ and $c^*$ are continuous by premise. Assumption~\ref{ass:ability_cost_regularity} makes $\Pi$ continuous in ability, and its uniform $1$-Lipschitz bound in $c^*$ gives joint continuity. Thus $\Phi$ is continuous. If $\Phi(\underline\alpha)\le0$ or $\Phi(\bar\alpha)\ge0$, the corresponding corner is self-consistent; otherwise the intermediate value theorem gives an interior root.
\end{proof}

This constructs a cutoff for continuous inner maps over the ability support at fixed stock; it does not establish active existence or stock continuity. Fixed-data single crossing signs only the first bracket in
\[
\Phi'(a)=[\nu\frac{\partial\Pi}{\partial\alpha}-\frac{\partial w}{\partial\alpha}]+\nu'\Pi+\nu F(c^*\mid a)c^{*\prime}.
\]
On a specified active cutoff interval, outer feedback dominance requires
\[
\nu'\Pi+\nu F(c^*\mid a)c^{*\prime}<\frac{\partial w}{\partial\alpha}-\nu\frac{\partial\Pi}{\partial\alpha},
\]
with this inequality integrating to strict decrease of $\Phi$ across the interval. Piecewise absolute continuity and a uniform positive margin suffice. The physical inactive lower residual $\Phi(\underline\alpha)=-w(\underline\alpha,K;\pi)$ can differ from the active right limit. Compact-domain conclusions below require dominance only on their certified interval; Proposition~\ref{prop:global_continuity} requires decrease over the full active-inner-cutoff domain, including its physical upper endpoint.

\subsection{Equilibrium on a compact domain}

Fix an environment $e$ and a compact archive interval $\mathcal K$. The inactive period equilibrium exists at every $K\in\mathcal K$. For active equilibria, let $\widehat\sigma(\sigma;a,K)$ be the inner-loop map defined above, and suppose there are a compact interval $A=[a_0,a_1]\subset(\underline\alpha,\bar\alpha]$ and $\varepsilon\in(0,1)$ such that:
\begin{enumerate}[(i)]
\item $\widehat\sigma$ and the objects generated by (E1)--(E7) are jointly continuous on $[\varepsilon,1]\times A\times\mathcal K$, and $\widehat\sigma(\varepsilon;a,K)>\varepsilon$ for every $(a,K)\in A\times\mathcal K$;
\item a continuous selection of active inner solutions on $A\times\mathcal K$ generates an outer residual $\Phi(a,K)$ that is jointly continuous and satisfies $\Phi(a_0,K)\ge0\ge\Phi(a_1,K)$ for every $K\in\mathcal K$.
\end{enumerate}

\begin{proposition}[Period-equilibrium existence, uniqueness, and continuity]
\label{prop:period_equilibrium}
Then an active period equilibrium exists at every $K\in\mathcal K$. If, in addition, the inner-loop uniqueness condition in Lemma~\ref{lem:uniqueness} holds and outer feedback dominance holds uniformly on $A\times\mathcal K$ (so $\frac{\partial\Phi}{\partial a}<0$ wherever the derivative exists, with the inequality integrating across any finitely many branch or clamp boundaries), then the active period equilibrium is unique within $A$ and every object in (E1)--(E8), including $h^e(K)$, is continuous in $K$ on $\mathcal K$.

\end{proposition}

\begin{proof}
For fixed $(a,K)$, continuity and the signs $\widehat\sigma(\varepsilon;a,K)>\varepsilon$ and $\widehat\sigma(1;a,K)\le1$ give an active inner root in $(\varepsilon,1]$. The assumed continuous active selection makes $\Phi$ continuous; its weak endpoint signs then give a cutoff in $A$. The generated objects satisfy (E1)--(E8).

Under Lemma~\ref{lem:uniqueness}, every subsequential limit of inner roots along $(a_n,K_n)\to(a,K)$ lies in $[\varepsilon,1]$ and solves the limiting equation. Uniqueness identifies it with the root at $(a,K)$, proving joint continuity. Outer feedback dominance similarly gives a unique cutoff in $A$; compactness and continuity of $\Phi$ make that cutoff continuous in stock by the same root-limit argument. All remaining objects are continuous compositions, including $h^e=\Delta(1-\pi)q_H^e\sigma^e$.
\end{proof}

\subsection{Sufficient conditions for global continuity}
\label{app:global_continuity}

Fix an environment and suppress its superscript. The following conditions certify continuity of selected creation for the dynamic results. Write $x$ for resolution, $a$ for the conjectured contributor cutoff, $A_\theta(K)=1-a_\theta^e(K)$, and
\[
q_\theta(x,K)=A_\theta(K)\Gamma(xV_\theta),\qquad
Q(x,K)=\pi q_L(x,K)+(1-\pi)q_H(x,K),\qquad N(a)=\Psi(a).
\]
On $Q,N>0$, define
\[
\begin{aligned}
H(x,a,K)&=\min\{1,N/Q\}\bar s(c(x,K),a),\\
R(x,a,K)&=\min\{1,Q/N\}\Pi(a;c(x,K))-w(a,K;\pi),\\
c(x,K)&=\max\{0,\beta\Delta(1-\pi)q_H(x,K)/Q(x,K)+u-C(K)\}.
\end{aligned}
\]
The remaining objects are exactly those in (E1)--(E7). An active equilibrium is a pair $(x,a)$ with $x\in(0,1]$, $Q,N>0$, $x=H(x,a,K)$, and $R=0$ when $a<\bar\alpha$, or $R\ge0$ when $a=\bar\alpha$. The lower ability endpoint cannot be active because $N(\underline\alpha)=0$. The formula for $R$ does not apply at this endpoint. Creation at an active pair is
$b(x,a,K)=\Delta(1-\pi)q_H(x,K)x$; inactive creation is zero.

\begin{assumption}[Sufficient full-domain continuity conditions]
\label{ass:global_continuity}
In addition to the baseline parameter bounds and participation interpretation, impose:
\begin{enumerate}[(P1)]
\item Primitive continuity. On the finite ability support, $\Psi$ is a continuous strictly increasing CDF with endpoint values zero and one. The cost CDF satisfies $F(0\mid\alpha)=0$ and is nondecreasing in ability. The functions $A_L,A_H$ and $C$ are continuous on $[0,\infty)$, and $w\ge0$ is jointly continuous in ability and stock. For every compact stock interval $J$, put $B_J=\max\{0,\beta\Delta+u-\min_J C\}$; $F$ is jointly continuous on $[0,B_J]\times[\underline\alpha,\bar\alpha]$, using its CDF extension beyond any finite cost support. The posting CDF is continuous and satisfies $\Gamma(0)=0$.
\item Inner density and gain conditions. The posting CDF is absolutely continuous with density $\gamma$ and $\Gamma(y)>0$ for $y>0$. Each conditional cost CDF is absolutely continuous with a supplied jointly Borel measurable density $f(c\mid\alpha)$, and
\[
\bar f=\sup_{c\ge0,\,a>\underline\alpha}
\frac{1}{\Psi(a)}\int_{\underline\alpha}^{a}f(c\mid\alpha)\,d\Psi(\alpha)<\infty.
\]
If $\beta\Delta\bar f>0$, the finite defined bound
\[
\kappa=\frac{\beta\Delta\bar f}{4}
\sup_{0<x<1}[\ell_H(x)-\ell_L(x)]^+<1,\qquad
\ell_\theta(x)=\frac{V_\theta\gamma(xV_\theta)}{\Gamma(xV_\theta)},
\]
holds. If $\beta\Delta\bar f=0$, use the separate nonincreasing-$H$ argument and set $\kappa=0$; no product of zero and an infinite supremum is required.
\item Full active-domain outer decrease. After (P1)--(P2) establish inner uniqueness, let $D(K)$ be the set of $a\in(\underline\alpha,\bar\alpha]$ admitting a positive inner root with $Q>0$, denote that root by $\xi(a,K)$, and set $\Phi(a,K)=R(\xi(a,K),a,K)$. For every $K\ge0$ and $a_1<a_2$ in $D(K)$, require $\Phi(a_2,K)<\Phi(a_1,K)$, including the physical upper endpoint whenever present.
\end{enumerate}
\end{assumption}

Condition (P3) concerns the full active-inner-cutoff domain, not only one compact interval $A$. A total-derivative outer-dominance inequality is sufficient only if it integrates on that domain. Fixed-data participation single crossing does not by itself imply (P3). No active existence at every stock or cutoff, or continuity of a selected branch, is assumed.

\begin{proposition}[Sufficient conditions for continuous selected creation]
\label{prop:global_continuity}
Under Assumption~\ref{ass:global_continuity}, there is at most one active equilibrium at each $K\ge0$, and active-first selected creation $h(K)$ is continuous on $[0,\infty)$, including activation, exit, weak upper participation corners, and zero-creation active equilibria.
\end{proposition}

\begin{proof}
Inner decrease and the cutoff domain. On any compact positive resolution interval with $Q>0$ and at any $a>\underline\alpha$, the density representations make $q_\theta,Q,\omega=(1-\pi)q_H/Q,c$, and $\mu$ absolutely continuous. Tonelli's theorem gives
$\bar s(c_2,a)-\bar s(c_1,a)\le\bar f(c_2-c_1)$ for $c_2\ge c_1$, so the pool average is nondecreasing and $\bar f$-Lipschitz in cost. Hence $H$ is absolutely continuous. At a common full-measure set of differentiation points, $\mu'\le0$ and
\[
\omega'=\omega(1-\omega)(\ell_H-\ell_L),\qquad
\frac{\partial H}{\partial x}\le\bar f(c')^+\le\beta\Delta\bar f(\omega')^+\le\kappa.
\]
The same formula holds when one type flow is zero, with $\omega\in\{0,1\}$. The Lipschitz inequality controls the composition with $c$ even when the pool average is not differentiable at the attained cost. In the zero-gain case, $c$ or the pool average is constant on the relevant range, and $H$ is nonincreasing. Integrating yields
\[
[H(x_2,a,K)-x_2]-[H(x_1,a,K)-x_1]
\le-(1-\kappa)(x_2-x_1)<0.
\]
Continuity extends this inequality to $x_2=1$. Thus the positive inner root is unique.

At fixed $(x,K)$, enlarging the contributor pool weakly raises its average answer probability: each added ability has $F(c\mid\alpha)$ at least as large as the old average. Both $\bar s$ and $\mu$ therefore weakly increase in $a$, so $H$ does too. If $a_1$ admits a root $x_1$, then for $a_2>a_1$,
$H(x_1,a_2,K)-x_1\ge0$ and $H(1,a_2,K)-1\le0$. Positive flow persists on $[x_1,1]$, and the intermediate value theorem supplies a positive root, including $x_1=1$. Thus $D(K)$ is upward closed and connected.

Active uniqueness. Two active equilibria with the same cutoff coincide by inner uniqueness. For two distinct cutoffs $a_1<a_2$, the lower is interior, so $\Phi(a_1,K)=0$. Condition (P3) makes $\Phi(a_2,K)<0$, contradicting either the interior equation or the weak upper-corner inequality.

Local persistence at positive creation. At any active inner root $(x_0,a_0,K_0)$, choose $0<x_-<x_0$ and $x_+>x_0$ if $x_0<1$, or $x_+=1$ otherwise. Some weighted private-failure rate is positive at $K_0$, so (P2) gives $Q(x_-,K_0)>0$. On a smaller relative stock neighborhood and an ability neighborhood bounded away from $\underline\alpha$, $Q,N$ remain uniformly positive. All generated objects are jointly continuous there by (P1), including at $N=Q$; continuity of the moving-limit pool integral follows from uniform continuity of $F$ and continuity of $\Psi$. Strict inner decrease gives a strict positive lower-face sign and a strict negative upper-face sign when $x_0<1$; for $x_0=1$ use $H(1,a,K)\le1$. These signs persist locally. The unique bracketed root $\xi(a,K)$ is jointly continuous: compactness supplies subsequential limits, the equation passes to each limit, and uniqueness identifies it. Every cutoff in this local rectangle is in $D(K)$, so its reduced residual is continuous and strictly decreasing.

For an interior equilibrium, bracket the zero of $\Phi$ by strict opposite signs on either side of $a_0$. For $a_0=\bar\alpha$, choose a lower cutoff $a_-<\bar\alpha$ with $\Phi(a_-,K_0)>\Phi(\bar\alpha,K_0)\ge0$. Preserve this strict lower sign nearby. Select $\bar\alpha$ if its residual is nonnegative; otherwise select the unique interior zero. In either case the local participation graph is closed under the original weak upper rule, and compactness and uniqueness make its cutoff continuous, including equality transitions. This proves a continuous local active branch. Positive creation persists by continuity of $b$, and full active uniqueness identifies the branch with the selection.

Zero-creation limits. Let $E_+(K)$ collect the active pairs and define
\[
\mathcal B(K)=\{0\}\cup\{b(x,a,K):(x,a)\in E_+(K)\}\subset[0,\Delta].
\] Under (P1), it has a closed graph on each compact stock interval. Indeed, if approaching active outputs have a positive limit, compactness supplies limits of $x$ and $a$; positive limiting creation forces $x,q_H,Q>0$. The bound $x\le N/Q$ rules out a zero limiting pool. Ordinary continuity then passes the inner equation and the interior or weak upper participation rule to the limit. A zero output limit belongs to $\mathcal B$ by definition. This also proves compactness of its values. If selected $h(K_0)=0$, active uniqueness implies $\mathcal B(K_0)=\{0\}$. Every subsequential limit of approaching selected outputs must therefore be zero. Together with local persistence at positive creation, this proves continuity of $h$, relative to the nonnegative stock domain at zero.
\end{proof}

\subsection{Intermediate viability}

Maintain the active-existence and uniqueness conditions of Proposition~\ref{prop:period_equilibrium} at $(K_0,\mathrm{AI})$ for some $K_0>0$, including outer feedback dominance on its certified interval $A$. Choose $\alpha_0\in A$ and define
\[
\begin{gathered}
c_0\equiv u-C(K_0)>0,\qquad N_0\equiv\Psi(\alpha_0)>0,\qquad
s_0\equiv F(c_0\mid\underline\alpha)>0,\qquad \sigma_0\equiv N_0s_0,\\
q_{H0}\equiv[1-a_H^{\mathrm{AI}}(K_0)]\,\Gamma(V_H\sigma_0)>0,\qquad
\nu_0\equiv\min\!\left\{1,\frac{(1-\pi)q_{H0}}{N_0}\right\},\qquad
\Pi_0\equiv\int_0^{c_0}F(c\mid\alpha_0)\,dc.
\end{gathered}
\]
\begin{lemma}[Sufficient condition for intermediate viability on a compact cutoff domain]
\label{lem:primitive_viability}
If
\[
\nu_0\Pi_0\ge w^{\mathrm{AI}}(\alpha_0,K_0;\pi)
\quad\text{and}\quad
\Delta(1-\pi)q_{H0}\sigma_0>\lambda K_0,
\]
then the unique active equilibrium with cutoff in $A$ satisfies $\phi^{\mathrm{AI}}(K_0)>\lambda$. This conclusion applies to an active-first selection when its cutoff belongs to $A$; active equilibria outside $A$ are not excluded.
\end{lemma}

\begin{proof}
First conjecture the cutoff $a=\alpha_0$, so contributor mass is $N_0$. Because each arriving query contributes at most one unit to posted flow, $Q\le1$. Hence $\mu=\min\{1,N_0/Q\}\ge N_0$. Also $\bar\Delta\ge0$ implies $c^*\ge c_0$, and the stochastic ordering of $F$ gives
\[
\bar s(c^*,\alpha_0)=\frac{1}{N_0}\int_{\underline\alpha}^{\alpha_0}F(c^*\mid\alpha)\,d\Psi(\alpha)
\ge F(c_0\mid\underline\alpha)=s_0.
\]
Thus the active inner solution has $\sigma=\mu\bar s\ge N_0s_0=\sigma_0$, so $q_H\ge q_{H0}$ and $Q\ge(1-\pi)q_{H0}$. It follows that $\nu\ge\nu_0$ and $\Pi(\alpha_0;c^*)\ge\Pi_0$. The first sufficient condition in the lemma therefore gives $\Phi(\alpha_0)\ge0$. Outer feedback dominance makes $\Phi$ strictly decreasing on $A$. Its endpoint signs yield a unique root there; since $\alpha_0\in A$ and $\Phi(\alpha_0)\ge0$, that root satisfies $\alpha^{*,\mathrm{AI}}(K_0)\ge\alpha_0$.

At the resulting equilibrium, contributor mass is at least $N_0$. Repeating the same bounds yields $\sigma^{\mathrm{AI}}(K_0)\ge\sigma_0$ and $q_H^{\mathrm{AI}}(K_0)\ge q_{H0}$. Therefore
\[
h^{\mathrm{AI}}(K_0)=\Delta(1-\pi)q_H^{\mathrm{AI}}(K_0)\sigma^{\mathrm{AI}}(K_0)
\ge\Delta(1-\pi)q_{H0}\sigma_0>\lambda K_0,
\]
which is equivalent to $\phi^{\mathrm{AI}}(K_0)>\lambda$.
\end{proof}

\section{Parametric Example}
\label{app:parametric}

Figures~\ref{fig:parametric} and~\ref{fig:eta_shifts} use a common, uncalibrated specification of the model. At each archive stock, we solve the period-equilibrium conditions, selecting the active equilibrium when one exists and zero creation otherwise. Ability and answering costs satisfy
\[
\alpha\sim\mathrm{Uniform}[1,2],\qquad
F(c\mid\alpha)=1-e^{-(8+\alpha)c},\quad c\ge0.
\]
The answering-cost shifter and posting-cost distribution are
\[
C(K)=\frac{2}{(1+K)^2},\qquad
\Gamma(d)=1-e^{-d/0.25},\quad d\ge0.
\]
Outside options are
\[
\begin{aligned}
w^{\mathrm{HO}}(\alpha,K)&=0.9(0.4+0.1\alpha)v(K),\\
w^{\mathrm{AI}}(\alpha,K)&=(0.5+0.1\alpha)v(K),\\
v(K)&=0.15+\frac{3K}{100+K}.
\end{aligned}
\]
Private-resolution rates are
\[
a_\theta^e(K)=b_\theta^e+\chi_\theta^e(1-e^{-K/2}),
\]
with coefficients
\begin{center}
\begin{tabular}{lcc}
\toprule
Environment & $(b_L,\chi_L)$ & $(b_H,\chi_H)$ \\
\midrule
Human-only & $(0.10,0.30)$ & $(0.05,0.15)$ \\
AI & $(0.30,0.60)$ & $(0.10,0.35)$ \\
\bottomrule
\end{tabular}
\end{center}
The remaining parameters are $\pi=0.5$, $\lambda=0.1$, $\beta=0.5$, $\Delta=1$, $u=0.8$, $V_H=1.5$, and $V_L=0.5$.

These parameters satisfy the maintained assumptions, the sufficient conditions for continuous selected creation, and intermediate viability. Code and supplementary verification for the numerical illustrations are available from the author upon request. Figure~\ref{fig:eta_shifts} uses conversion rates $\eta=0$, $0.25$, and $0.93$.

\section{Model Extensions}

\subsection{Longer query lifetimes}
\label{app:T_extension}

For finite $T>1$, unresolved queries remain active until answered or expiry. Matching uses the active stock $X^e$, rather than posted inflow $Q^e$. The analysis fixes $(K,N,Q,s)$ while a query ages, then solves stationary posting where specified. It does not describe fully nonstationary archive and hazard changes across cohorts.

Matches are uniform over the active stock and independent of query age. A contributor decides once and is not rematched to a query they declined; with a continuum of contributors, excluding finitely many past matches does not change the ability distribution. The baseline answering rule and pool average \eqref{eq:sbar} therefore give the age-independent hazard
\begin{equation}
\rho^e=\mu^e\bar s^{\,e},\qquad
\mu^e=\min\{1,\Psi(\alpha^{*,e})/X^e\}.
\label{eq:rho_general}
\end{equation}
At $X^e=0$, set $\mu^e=\rho^e=0$. Committed-slot timing remains as in Section~\ref{sec:accumulation}; contributor matching is $\nu^e=\min\{1,X^e/\Psi(\alpha^{*,e})\}$ with the corners in Section~\ref{sec:surplus}. A poster receives $V_\theta$ if resolved within the lifetime, without within-lifetime discounting, so $U_\theta^e=\sigma^e V_\theta$ for lifetime resolution $\sigma^e$.

\subsubsection{Stationary matching and posting}

Fix integer $T\ge1$, contributor mass $N\ge0$, inflow $Q>0$, and per-match answer probability $s\in[0,1]$. Write
\[
H_T(\rho)=\sum_{a=0}^{T-1}(1-\rho)^a,\qquad f_T(\rho)=1-(1-\rho)^T.
\]
Cohort accounting gives $X=QH_T(\rho)$ and $\sigma=f_T(\rho)$. The undivided selector $X_{\mathrm{unc}}=QH_T(s)$ is the stock at full hazard $\rho=s$.

\begin{proposition}[Stationary matching block under general $T$]
\label{prop:stationary_T}
Under this fixed block:
\begin{enumerate}[(i)]
\item The stationary hazard and stock exist uniquely, with
\begin{equation}
\sigma=\min\{1-(1-s)^T,Ns/Q\},\qquad X\rho=Q\sigma.
\label{eq:sigma_T_closed}
\end{equation}
\item $\sigma_T(N,Q,s)$ is nondecreasing in $N$, $s$, and $T$ and nonincreasing in $Q$. It is
strictly increasing in $N$ wherever $1-(1-s)^T>Ns/Q$; strictly decreasing in $Q$ wherever, in
addition, $N>0$; and independent of $T$ wherever $1-(1-s)^T>Ns/Q$ (strict congestion at $T$
implies strict congestion at every larger lifetime).
\item Let $N>0$, $s\in(0,1]$, and let posting respond to lifetime resolution through a
continuous nondecreasing $Q:[0,1]\to[0,\infty)$ with $Q(0)=0$ and $Q(\sigma)>0$ for
$\sigma>0$. Besides the inactive outcome ($\sigma=0$, $Q=0$, $X=0$), there is exactly one
active stationary outcome $\sigma^*\in(0,1]$, solving
$\sigma^*=\min\{1-(1-s)^T,\,Ns/Q(\sigma^*)\}$; it is nondecreasing in $N$, $s$, and $T$, and
nonincreasing under pointwise increases of $Q(\cdot)$, while the realized flow
$Q^*=Q(\sigma^*)$ is nondecreasing under pointwise increases of $Q(\cdot)$.

\end{enumerate}
\end{proposition}

Proof of parts (i)--(ii). For brevity write $H=H_T$ and $f=f_T$. The telescoping identity $\rho H(\rho)=f(\rho)$ holds even at zero; $f$ is continuous and strictly increasing, while $H$ is continuous, nonincreasing, and at least one. Stationary states are exactly the roots of
\[
\rho=\min\{1,N/[QH(\rho)]\}s,\qquad \rho\in[0,s].
\]
If $s=0$, the unique state is $\rho=\sigma=0$, $X=QT$, and $\mu=\min\{1,N/(QT)\}$; the physical selector is $QT\le N$.

For $s>0$, an uncongested state forces $\rho=s$ and $X=X_{\mathrm{unc}}$, feasible exactly when $X_{\mathrm{unc}}\le N$. Otherwise $0\le Ns/Q<f(s)$, and the congested equation $f(\rho)=Ns/Q$ has exactly one root in $[0,s)$. Its stock $QH(\rho)\ge X_{\mathrm{unc}}>N$ confirms congestion. When $N>0$, $\rho>0$, $X=Ns/\rho$, and $\nu=1$; when $N=0$, $\rho=0$ and $X=QT$ without division by the hazard. At $X_{\mathrm{unc}}=N$, both formulas give $\rho=s$ and $X=N$. The branches are exhaustive and agree at their boundary. They yield the minimum formula, while $X\rho=Qf(\rho)=Q\sigma$ holds in every case.

The minimum inherits weak increase in $N,s,T$ and weak decrease in $Q$. If $f(s)>N_1s/Q$, increasing $N$ makes both entries exceed $N_1s/Q$, even if the larger pool leaves congestion. Increasing $Q$ strictly lowers resolution on this branch when $N>0$; with $N=0$, resolution stays zero. Finally, $f_{T'}(s)\ge f_T(s)>Ns/Q$ for every $T'\ge T$, so strict congestion persists and resolution remains $Ns/Q$. \qed

Proof of part (iii). For fixed $N>0$, $s>0$, write $f=f_T$ and define on $(0,1]$
\[
D(\sigma)=\min\{f(s),Ns/Q(\sigma)\}-\sigma.
\]
Continuity and monotonicity of $Q$ make $D$ continuous and strictly decreasing. Since $Q(0)=0$, its limit at zero is $f(s)>0$, while $D(1)\le0$. The intermediate value theorem gives exactly one positive root. Part (i) recovers its unique hazard-stock pair. The inactive convention ($\sigma=Q=X=0$) coexists.

Higher $N,s,T$ raise the minimum pointwise and therefore weakly raise its unique crossing with $\sigma$. A pointwise higher posting schedule lowers that crossing. To compare realized flows, let $\tilde Q\ge Q$ with $\tilde\sigma^*\le\sigma^*$. If capacity binds at the new outcome, including equality,
\[
\tilde Q(\tilde\sigma^*)=Ns/\tilde\sigma^*\ge Ns/\sigma^*\ge Q(\sigma^*).
\]
Otherwise the new lifetime bound binds strictly, so $\tilde\sigma^*=f(s)$. The ordering and the old lifetime bound force $\sigma^*=f(s)$ too, and $\tilde Q(f(s))\ge Q(f(s))$. Thus realized flow weakly rises in either case. Uniqueness here fixes the contributor pool and answer rate, rather than solving their equilibrium feedbacks. \qed

\subsubsection{Accounting and congestion}

Stationary resolved flow is $X\rho=Q\sigma$, giving creation
\[
h^{\mathrm{AI}}(K;\eta)=\Delta(1-\pi)[q_H^{\mathrm{AI}}(K)\sigma^{\mathrm{AI}}(K)+\eta a_H^{\mathrm{AI}}(K)]
\]
with lifetime resolution $\sigma$. At $T=1$, \eqref{eq:sigma_T_closed} reduces to $\sigma=\min\{1,N/Q\}s$, the baseline.

Surviving cohorts compete for contributor capacity: $X\ge Q$, with equality at $T=1$ or, for $T\ge2$, at $\rho=1$. On strictly congested positive-resolution interiors, resolution has unit elasticity with respect to $N$; on uncongested interiors it has zero elasticity. Elasticity is undefined at zero resolution and kinked at the branch boundary. Once capacity binds, longer lifetimes leave resolution unchanged at fixed $(N,Q,s)$.

For $N>0$, uncongested contributor matching $\nu=X_{\mathrm{unc}}/N$ rises weakly with $T$, with increment $Q(1-s)^T/N$ when both lifetimes remain uncongested, and is capped at one under congestion. Under committed-slot timing this weakly raises the participation payoff at fixed block and match surplus.

\begin{corollary}[Resolution comparison under general $T$, congested branch]
\label{cor:race_T}
Fix a lifetime $T\ge1$ and suppose both environments are stock-congested at $K$ under that
lifetime ($1-(1-\bar s^{\,e})^T>\Psi(\alpha^{*,e})\,\bar s^{\,e}/Q^e$ for both $e$), with
positive contributor mass, both resolving, and the anchor positivity of
Proposition~\ref{prop:resolution}. Then $\sigma^e=\Psi(\alpha^{*,e})\,\bar s^{\,e}/Q^e$
exactly, and the four-factor decomposition \eqref{eq:resolution_race} holds at that $T$ with the
same factors and signs. Holding the stationary block objects $(\Psi(\alpha^{*,e}),Q^e,\bar s^{\,e})$ fixed at $K$ in each environment, stock congestion persists for every $T'\ge T$, so the factorization continues to hold at those lifetimes.
\end{corollary}

\begin{proof}
Proposition~\ref{prop:stationary_T} gives $\sigma^e=N^e s^e/Q^e$ and persistence of strict congestion at all larger lifetimes when block objects are fixed. The factorization and conditional signs in the proof of Proposition~\ref{prop:resolution} then apply under the same positive-anchor premise.
\end{proof}
Off joint congestion, the full minimum in \eqref{eq:sigma_T_closed}, including its lifetime term, governs comparison.

\subsection{A general matching technology}
\label{app:matching_general}

The baseline matches queries and contributors by the short-side rule: $\min\{A,N\}$ matches
between an active mass $A$ ($A=Q$ under $T=1$; $A=X$ in Appendix~\ref{app:T_extension}) and a
contributor mass $N=\Psi(\alpha^*)$. This subsection replaces the rule with a matching
technology $M(N,A)$ and records which conclusions survive. Maintain: $M$ is continuous,
nondecreasing in each argument, $M(N,A)\le\min\{N,A\}$, $M(0,A)=M(N,0)=0$; and the per-query
match rate $\mu=M(N,A)/A$ is nonincreasing in $A>0$ (no increasing returns to queue length). The
short-side rule is the frictionless limiting case.

Directional comparisons and inner uniqueness. The hazard is $\rho=[M(N,A)/A]\bar s$, equal to $\sigma$ at $T=1$ and increasing lifetime resolution through $\sigma=f_T(\rho)$ otherwise. Pool thinning weakly lowers matching, strictly where $\frac{\partial M}{\partial N}>0$; lower active stock weakly raises per-query matching under the ratio restriction. Pool quality and composition act through $\bar s$. The exact four-factor product in \eqref{eq:resolution_race} uses the short-side form. At $T=1$, with the density, differentiability, and absolute-continuity conditions of Lemma~\ref{lem:uniqueness}, the matching-rate derivative is weakly nonpositive and the same composition bound \eqref{eq:uniqueness_cond} gives inner uniqueness.

Fixed-inflow stationary matching. For every integer $T\ge1$, the stationary matching block has a unique solution at
fixed $Q>0$, $N\ge0$, and $s=\bar s\in[0,1]$. Use the cohort functions $H_T,f_T$ above. With $X(\rho)=QH_T(\rho)>0$, the stationary hazard equation is equivalent to
\[
E_Q(\rho)\equiv Qf_T(\rho)-sM(N,QH_T(\rho))=0,
\qquad \rho\in[0,s].
\]
For $s>0$, $E_Q$ is continuous and strictly increasing: $f_T$ is strictly increasing,
$H_T$ is nonincreasing, and total matching $M(N,\cdot)$ is nondecreasing. Its endpoint signs are
\[
E_Q(0)=-sM(N,QT)\le0,\qquad
E_Q(s)=s[QH_T(s)-M(N,QH_T(s))]\ge0.
\]
The intermediate value theorem gives exactly one hazard, hence one stock and one lifetime
resolution probability. This proof uses total-matching monotonicity, without contraction or the per-query ratio restriction.

If $s=0$, the solution is $\rho=\sigma=0$, $X=QT$ and $\mu=M(N,QT)/(QT)$.
If $N=0$, the same stock remains but $\mu=0$. More generally, $\rho=0$ exactly when
$sM(N,QT)=0$. If $Q=0$, cohort accounting gives $X=0$, and the separate empty-stock convention
sets $\mu=\rho=\sigma=0$. In every case with $Q>0$,
\[
Q\sigma=X\rho=sM(N,X).
\]
The short-side closed form and capacity-pinned throughput do not generally survive:
$X>N$ need not imply $M(N,X)=N$.

Restricted endogenous posting. Keep $N,s,T$ fixed and write $\Sigma(q)$
for lifetime resolution in the unique matching block at inflow $q>0$. Joint continuity of
$E_q$ and uniqueness of its root on the compact interval $[0,s]$ make $\Sigma$ continuous:
every subsequential hazard limit along $q_n\to q>0$ solves the limiting residual,
so compactness and uniqueness imply convergence to its root. The maintained ratio restriction gives
\[
\frac{E_q(\rho)}q=f_T(\rho)-sH_T(\rho)
\frac{M(N,qH_T(\rho))}{qH_T(\rho)},
\]
which is nondecreasing in $q$ at every fixed $\rho$. Comparing the strictly increasing
residuals' roots shows that $\Sigma(q)$ is nonincreasing. For continuous nondecreasing
$Q:[0,1]\to[0,\infty)$ with $Q(0)=0$ and $Q(\sigma)>0$ for $\sigma>0$, the residual
$D(\sigma)=\Sigma(Q(\sigma))-\sigma$ is continuous and strictly decreasing on $(0,1]$.
Thus there is at most one active stationary outcome.

For existence, suppose $s>0$ and $M(N,x_0)>0$ for some $x_0>0$. Stock monotonicity and the
ratio restriction imply positive matching for every positive stock and
$m_0(N)\equiv\lim_{x\downarrow0}M(N,x)/x\in(0,1]$. Since $q\le X(q)\le qT$,
$\rho(q)\to sm_0(N)$ as $q\downarrow0$, so
$D(\sigma)\to f_T(sm_0(N))>0$ as $\sigma\downarrow0$, while $D(1)\le0$.
There is exactly one positive active outcome. If $s=0$ or $M(N,\cdot)$ vanishes identically,
there is no active outcome. The inactive profile is a separate convention: under the stated
nondegeneracy its zero hazard differs from the positive limiting hazard of the active block
as inflow tends to zero. These results hold at fixed $N,s,T$ and do not establish full-equilibrium uniqueness or stability.

\end{document}